%% file: ArXiv_version.tex
\documentclass[11pt]{article}

\usepackage[utf8]{inputenc}

\usepackage[ruled,vlined,linesnumbered,noend]{algorithm2e}
\SetKwInOut{Parameter}{Parameters}
\usepackage{amsthm}
\usepackage{amsmath}
\usepackage{amssymb}
\usepackage{enumerate}
\usepackage{graphicx}
\usepackage[margin=1in]{geometry}
\usepackage{dsfont}
\usepackage{booktabs} 
\usepackage{tabularx}
\usepackage{multirow}
\usepackage[shortlabels]{enumitem}
\usepackage{tablefootnote}

\usepackage[dvipsnames]{xcolor}
\definecolor{ForestGreen}{rgb}{0.1333,0.5451,0.1333}
\definecolor{DarkRed}{rgb}{0.65,0,0}

\usepackage{hyperref}
\hypersetup{
    colorlinks=true,
    linkcolor=DarkRed,
    filecolor=magenta,      
    urlcolor=cyan,
    citecolor=ForestGreen,
}

\usepackage{cleveref}

\usepackage{framed}
\usepackage[framemethod=tikz]{mdframed}
\usepackage{tikz-cd}
\newenvironment{wrapper}[1]
{
	\begin{center}
		\begin{minipage}{\linewidth}
			\begin{mdframed}[hidealllines=true, backgroundcolor=gray!20, leftmargin=0cm,innerleftmargin=0.4cm,innerrightmargin=0.4cm,innertopmargin=0.4cm,innerbottommargin=0.4cm,roundcorner=3pt]
				#1}
			{\end{mdframed}
		\end{minipage}
	\end{center}
}

\newcommand{\brac}[1]{\left(#1\right)}
\newcommand{\miss}{\mathsf{miss}}
\newcommand{\clr}{\mathsf{clr}}

\newcommand{\lst}{\mathsf{lst}}

\newcommand{\rdy}{\mathsf{rdy}}
\newcommand{\scl}{\mathsf{scl}}
\newcommand{\ind}{\mathsf{ind}}
\newcommand{\lon}{\mathsf{lon}}

\newcommand{\Vizing}{\textnormal{\textsf{Vizing}}}
\newcommand{\TVizing}{\textnormal{\textsf{TruncatedVizing}}}
\newcommand{\VizingF}{\textnormal{\textsf{VizingFan}}}

\DeclareMathAlphabet{\mathmybb}{U}{bbold}{m}{n}

\newtheorem{theorem}{Theorem}[section]
\newtheorem{lemma}[theorem]{Lemma}

\newtheorem{corollary}[theorem]{Corollary}

\newtheorem{definition}[theorem]{Definition}

\newtheorem{assumption}[theorem]{Assumption}
\newtheorem{observation}[theorem]{Observation}
\newtheorem{claim}[theorem]{Claim}

\title{Even Faster $(\Delta + 1)$-Edge Coloring \\via Shorter Multi-Step Vizing Chains}

\author{
	Sayan Bhattacharya$^*$ \and Mart\'in Costa$^*$ \and Shay Solomon\textsuperscript{\textdagger} \and Tianyi Zhang\textsuperscript{\textdagger}
}

\date{
University of Warwick$^*$ \\ Tel Aviv University\textsuperscript{\textdagger}
}

\begin{document}

\maketitle

\pagenumbering{gobble}

\begin{abstract}
\input{abstract}
\end{abstract}

\newpage

\tableofcontents

\newpage

\pagenumbering{arabic}

\part{Extended Abstract}

\input{intro}

\input{overview}
\newpage

\part{Full Version}\label{part:full version}

\input{prelim}

\input{tools}
\input{alg}
\input{extension}

\input{stars}

\input{ack}


\bibliography{bibliography.bib}
\bibliographystyle{alpha}

\end{document}

%% file: abstract.tex
Vizing's Theorem from 1964 states that 
any $n$-vertex $m$-edge graph with maximum degree $\Delta$ can be {\em edge colored} using at most $\Delta + 1$ colors.
For over 40 years, the state-of-the-art running time for computing such a coloring, obtained independently by Arjomandi [1982] and by Gabow,
Nishizeki, Kariv, Leven and Terada~[1985],
was $\tilde O(m\sqrt{n})$.
Very recently, this time bound was improved in two independent works, by Bhattacharya, Carmon, Costa, Solomon and Zhang to $\tilde O(mn^{1/3})$, and by
Assadi to $\tilde O(n^2)$.

In this paper we present an algorithm that computes such a coloring in $\tilde O(mn^{1/4})$ time.
Our key technical contribution is a subroutine for extending the coloring to one more edge within time $\tilde O(\Delta^2 + \sqrt{\Delta n})$. The best previous time bound of any color extension subroutine is either the trivial $O(n)$, dominated by the length of a Vizing chain, or the bound 
$\tilde{O}(\Delta^6)$ 
by Bernshteyn [2022], dominated by the length of {\em multi-step Vizing chains}, which is basically a concatenation of multiple (carefully chosen) Vizing chains.
Our color extension subroutine produces significantly shorter multi-step Vizing chains than in previous works, for sufficiently large $\Delta$.

%% file: intro.tex
\section{Introduction}
Let $G = (V, E)$ be a simple, undirected $m$-edge $n$-vertex graph with maximum degree $\Delta$. 
For an integer $\kappa \in \mathbb{N}^+$,
a $\kappa$-edge coloring $\chi : E \rightarrow \{1, 2, \ldots, \kappa\}$ of the graph $G$ assigns a {\em color} $\chi(e)$ to every edge $e \in E$, so that any two adjacent edges receive distinct colors. 
The minimum  $\kappa$ for which the graph $G$ admits a $\kappa$-edge coloring, or the {\em edge chromatic number} of $G$, cannot be smaller than $\Delta$. On the other hand, Vizing's Theorem states that  $\Delta+1$ colors are always  sufficient~\cite{Vizing}.

Vizing's original proof for the existence of $(\Delta+1)$-edge coloring can rather easily be converted into  an $O(m n)$ time algorithm. In two independent works from the 80s, Arjomandi~\cite{arjomandi1982efficient} and Gabow et al.~\cite{gabow1985algorithms} improved this runtime bound to $\tilde{O}(m\sqrt{n})$.\footnote{We use the notation $\tilde{O}(\cdot)$
throughout to suppress polylogarithmic in $n$ factors.} Recently, Sinnamon~\cite{sinnamon2019fast}  used randomization to achieve a clean bound of $O(m \sqrt{n})$. 
There has been  no polynomial improvement over this $m\sqrt{n}$ time barrier in over 40 years, until very recently, where this time bound was improved in two independent works. Bhattacharya, Carmon, Costa, Solomon and Zhang \cite{BhattacharyaCCSZ24} improved the time bound to $\tilde O(mn^{1/3})$, which provides a polynomial improvement over the $m\sqrt{n}$ time bound in the entire regime of parameters. \cite{Assadi24} achieved a time bound of $\tilde{O}(n^2)$, which in particular provides a near-linear time algorithm for dense graphs; refer to \Cref{related} for additional results in \cite{Assadi24}, where more than $\Delta+1$ colors are used.
Both algorithms of \cite{BhattacharyaCCSZ24} and \cite{Assadi24} are randomized and the running time bounds hold with high probability.
Remarkably, the approaches in \cite{BhattacharyaCCSZ24}
and in \cite{Assadi24} are inherently different. In a nutshell, the key contribution of \cite{BhattacharyaCCSZ24} is in speeding up the coloring of the last few edges whereas the key contribution of \cite{Assadi24} is in speeding up the coloring of all but the last few edges. 

Despite this exciting recent progress,
the following outstanding question remains open.
\begin{wrapper}
How fast can one compute a $(\Delta+1)$-edge coloring of an input $m$-edge $n$-vertex graph?
\end{wrapper}
By combining the results of \cite{BhattacharyaCCSZ24} 
and \cite{Assadi24}, one can directly get a running time of $\tilde{O}(n \sqrt{m})$.\footnote{This observation was made via personal communication with the author of \cite{Assadi24}.}
Other than this direct corollary of the combination of the two works \cite{BhattacharyaCCSZ24,Assadi24}, it is unclear whether any further improvement is possible: This is exactly where the contribution of the current paper lies.
Specifically, 
we\footnote{{\em quasi nanos, gigantium humeris insidentes}} prove the following theorem.

\begin{theorem}\label{thm:main}
Given a simple, undirected $m$-edge $n$-vertex graph $G = (V, E)$ with maximum degree $\Delta$, we can compute a $(\Delta + 1)$-edge coloring of $G$ in $\tilde O(mn^{1/4})$ time with high probability.
\end{theorem}

\noindent
Our proof of this theorem
builds on the recent works of \cite{BhattacharyaCCSZ24,Assadi24}, as well as on a sequence of earlier works \cite{duan2019dynamic,Bernshteyn22,Christiansen23}.
Our key technical contribution is a subroutine for extending a partial coloring to one more edge within time $\tilde O(\Delta^2 + \sqrt{\Delta n})$. The best previous time bound of any color extension subroutine is either the trivial $O(n)$, dominated by the length of a Vizing chain, or the bound 
$\tilde{O}(\Delta^6)$ 
by Bernshteyn~\cite{Bernshteyn22}, dominated by the length of {\em multi-step Vizing chains}, which is basically a concatenation of multiple (carefully chosen) Vizing chains. 
There is also a much faster color extension subroutine, by Duan et al.\ \cite{duan2019dynamic}: it uses $(1+\epsilon)\Delta$ colors within time $\tilde O(1/\epsilon^2)$, provided that 
$\Delta = \tilde{\Omega}(1)$ and
$\epsilon = \tilde \Omega(1/\sqrt{\Delta})$; we note that the minimum number of colors achievable by \cite{duan2019dynamic} is $\Delta + \tilde{O}(\sqrt{\Delta})$,
and the respective runtime is  $\tilde{O}(\Delta)$.
Even slightly reducing the number of colors below 
$\Delta + \tilde{O}(\sqrt{\Delta})$, while allowing a higher time of $\tilde{O}(\Delta^2)$, 
is currently out of reach.

Our color extension subroutine settles for the aforementioned higher running time of  $\tilde O(\Delta^2 + \sqrt{\Delta n})$; even then, it
has to drill much deeper than \cite{duan2019dynamic}, since the transition from 
 $\Delta + \tilde{O}(\sqrt{\Delta})$ colors to $\Delta+1$ has to overcome numerous nontrivial technical hurdles.
  Ultimately, we manage to
produce significantly shorter multi-step Vizing chains than in previous works for $\Delta+1$ colors, in the regime that  $\Delta$ is sufficiently large; we then demonstrate the usefulness of such a color extension subroutine in proving \Cref{thm:main}.
A comprehensive technical overview is given in  \Cref{sec:overview}.

\medskip
\noindent \textbf{Subsequent Work:} Very recently, taking a completely different approach to the recent works of \cite{BhattacharyaCCSZ24, Assadi24}, \cite{assadi2024vizing} showed how to compute a $(\Delta + 1)$-edge coloring in $O(m \log \Delta)$ time with high probability, giving a near-linear time algorithm for this problem.

\subsection{Related Work} \label{related}

If we allow for (sufficiently) more colors than $\Delta+1$, then the problem becomes (dramatically) simpler.  First, it was known since the 80s that, for
$\Delta+\tilde{O}(\sqrt{\Delta})$ colors, the problem can be solved in $\tilde{O}(m)$ time~\cite{karloff1987efficient}.
A recent line of works provided algorithms with near-linear runtime for $(1+\epsilon)\Delta$-edge coloring~\cite{duan2019dynamic,BhattacharyaCPS24,ElkinK24}, culminating with the recent result of \cite{Assadi24} for $(\Delta + O(\log n))$-edge coloring in $O(m \log \Delta)$ time and a $(1+\epsilon)\Delta$-edge coloring in $O(m \log(1/\epsilon))$ time for any $\epsilon = \omega(\log n / \Delta)$.

  There is a large body of work on edge coloring in restricted graph classes. 
First, for bipartite graphs one can compute a $\Delta$-edge coloring in $\tilde{O}(m)$ time~\cite{cole1982edge,combinatorica/ColeOS01,alon2003simple}. One can  compute a $(\Delta+1)$-edge coloring in $\tilde{O}(m \Delta)$ time \cite{gabow1985algorithms}
for bounded degree graphs. 
This result of \cite{gabow1985algorithms}
 from the 80s was generalized recently for
  bounded arboricity graphs~\cite{BhattacharyaCPS24b}, and there has been further recent work on edge coloring in bounded arboricity graphs
\cite{BhattacharyaCPS24c,ChristiansenRV24,Kowalik24}.
Refer to \cite{chrobak1989fast,chrobak1990improved,cole2008new} and the references therein
for works on edge coloring in planar graphs, bounded treewidth graphs and bounded genus graphs.
 
Finally, the edge coloring problem is receiving an extensive and growing research attention in various computational models other than the classic static sequential setting, including in dynamic algorithms~\cite{BarenboimM17,BhattacharyaCHN18,duan2019dynamic,Christiansen23,BhattacharyaCPS24,Christiansen24}, distributed algorithms~\cite{panconesi2001some,elkin20142delta,fischer2017deterministic,ghaffari2018deterministic,balliu2022distributed,ChangHLPU20,Bernshteyn22,Christiansen23,Davies23}, online algorithms~\cite{CohenPW19,BhattacharyaGW21,SaberiW21,KulkarniLSST22,BilkstadSVW24,BlikstadOnline2025, dudeja2024randomizedgreedyonlineedge}, and streaming algorithms~\cite{BehnezhadDHKS19,behnezhad2023streaming,chechik2023streaming,ghosh2023low}.

\subsection{Organization of the Rest of the Paper}

To prove \Cref{thm:main}, we use 3 different subroutines, which we formally describe in \Cref{part:full version} of our paper. In \Cref{sec:overview}, we give a technical overview of our algorithm and analysis. In Sections~\ref{sec:lm:main} and \ref{sec:th:extension:main}, we sketch the analysis of the subroutines that we use to prove \Cref{thm:main}. In \Cref{part:full version}, we provide the full version of our paper.\footnote{See the first paragraph in \Cref{part:full version} for the organization of the full version of our paper.}

%% file: overview.tex
\section{Overview of Our Techniques}
\label{sec:overview}

Our key contributions are to design an efficient algorithm for the following subroutine, and to show that it leads to Theorem~\ref{thm:main}  in combination with the machinery developed in~\cite{BhattacharyaCCSZ24} and~\cite{Assadi24}.

\medskip
\noindent {\bf $(\Delta+\eta)$-Color Extension Subroutine:} Here, $\eta \geq 1$ is an integer. The input to this subroutine is a graph $G = (V, E)$ with maximum degree $\Delta$, a given edge $e \in E$, and a valid coloring $\chi_{\texttt{init}} : E \setminus \{e \} \rightarrow [\Delta+\eta]$ of $G \setminus \{e\}$. The subroutine has to {\em extend} the {\em partial} coloring $\chi$ to the entire graph $G$, by assigning some color $c \in [\Delta+\eta]$ to the uncolored edge $e \in E$ and possibly changing the colors of some edges in $E \setminus \{e \}$. Let $\chi_{\texttt{final}} : E \rightarrow [\Delta+\eta]$ be the valid $(\Delta+\eta)$-edge coloring of $G$ when the subroutine finishes execution. We define the {\bf cost} incurred by the subroutine to be the number of edges in $G$ that change their colors during this process, i.e., the cost equals $1+ | \{ e' \in E \setminus \{e\} : \chi_{\texttt{init}}(e') \neq \chi_{\texttt{final}}(e') \}|$. It is easy to observe that {\bf the runtime of any such subroutine is at least its cost}, because the subroutine needs to spend $\Omega(1)$ time to change the color of any given edge. We now summarize a lower bound derived by Chang et al.~\cite{ChangHLPU20}.

\begin{theorem}
[\cite{ChangHLPU20}]
\label{thm:lowerbound} Any $(\Delta+\eta)$-color extension subroutine  has  cost $\Omega((\Delta/\eta) \log (\eta n/\Delta))$.
\end{theorem}

\medskip
\noindent {\bf Why Is This Useful?} Suppose that we had a $(\Delta+1)$-color extension subroutine whose runtime matched the lower bound of Theorem~\ref{thm:lowerbound}, up to polylogarithmic factors. Then this would imply an algorithm for $(\Delta+1)$-edge coloring that runs in $\tilde{O}(m)$ time! Below, we explain this in more detail.  

We start by applying a well-known  {\em Eulerian partition} technique~\cite{arjomandi1982efficient,gabow1985algorithms,sinnamon2019fast}. In $\tilde{O}(m)$ time, this allows us to partition the input graph $G = (V, E)$ into two (almost) equal-sized subgraphs $G_1 = (V, E_1)$ and $G_2 = (V, E_2)$, such that maximum degee in each of these subgraphs is  $\leq \lceil \Delta/2 \rceil$. We then recursively color $G_1$ and $G_2$   using two mutually {\em disjoint} palettes, each of size $\lceil \Delta/2 \rceil+1$. This gives us a $2 \cdot (\lceil \Delta/2 \rceil + 1) \leq (\Delta+3)$-edge coloring of $G$. We next uncolor those edges in $G$ that belong to the two {\em least popular} color classes. By a simple averaging argument, this leads  to a partial $(\Delta+1)$-edge coloring of $G$ with $O(m/\Delta)$ uncolored edges. We next scan through these uncolored edges, and extend the partial $(\Delta+1)$-edge coloring to them one at a time by applying the $(\Delta+1)$-color extension subroutine. Note that we make $O(m/\Delta)$ calls to the color extension subroutine, each call taking $\tilde{O}(\Delta)$ time. Thus, the total time spent on all these calls is $O(m/\Delta) \cdot \tilde{O}(\Delta) = \tilde{O}(m)$. At the end of the scan, we have a $(\Delta+1)$-edge coloring of $G$. The overall runtime  is captured by the following recurrence, whose solution is $T(m, \Delta) = \tilde{O}(m)$.
\begin{equation}
\label{eq:recurrence}
T(m, \Delta) = 2 \cdot T\left( \lceil m/2 \rceil, \lceil \Delta/2 \rceil  \right) + \tilde{O}(m).
\end{equation}
Accordingly, a natural line of attack on the problem is to design an efficient  $(\Delta+1)$-color extension subroutine. For this plan to work, however, we first need to address the following concern.

 What happens if we have a $(\Delta+1)$-color extension subroutine that runs in, say, $\tilde{O}(\Delta^2)$ time (which is reasonably fast, but falls short of matching the lower bound of Theorem~\ref{thm:lowerbound})? Unfortunately, if we plug in such a color extension subroutine in the framework described above, then we spend $\tilde{O}(\Delta^2) \cdot O(m/\Delta) = \tilde{O}(m\Delta)$ total time  to extend the coloring to the last $O(m/\Delta)$ uncolored edges. So, the last term in the RHS   of~(\ref{eq:recurrence})  becomes $\tilde{O}(m \Delta)$ instead of  $\tilde{O}(m)$. Thus, using this approach, the time taken to compute a $(\Delta+1)$-edge coloring  is now  $\Omega(m \Delta)$, which is no better than the existing $\tilde{O}(m \Delta)$ time algorithms by~\cite{arjomandi1982efficient,gabow1985algorithms,sinnamon2019fast}. So the concern is that the above framework is not {\em robust}: It does not give any {\em polynomial advantage} over the current state-of-the-art, if our $(\Delta+1)$-color extension subroutine takes $\Omega(\Delta^2)$ time.  To address this concern, we synthesize the main technical insights   from~\cite{BhattacharyaCCSZ24}, and prove the following lemma.  

 \begin{lemma}
\label{lm:key}
Consider a partial coloring $\chi : E \setminus E^{\star} \rightarrow [\Delta+1]$ of a graph $G = (V, E)$, where $E^{\star} \subseteq E$ is the set of uncolored edges. Let  $U^{\star} \subseteq V$ be a vertex cover of $G^{\star} := (V, E^{\star})$. Then there is a randomized algorithm that {\em extends} $\chi$ to $\Omega\left(|E^{\star}|\right)$ edges of $E^{\star}$ in $\tilde{O}\left( |E^{\star}|  \Delta + \Delta  m  |U^{\star}|/|E^{\star}| \right)$ expected time. Specifically, the algorithm terminates with a partial coloring $\chi' : E \setminus E' \rightarrow [\Delta+1]$ of $G$, with $E' \subseteq E$ being the set of uncolored edges, such that $E' \subseteq E^{\star}$ and $|E^{\star} \setminus E'| = \Omega\left( |E^{\star}|\right)$.
\end{lemma}

The above lemma shows how to efficiently extend a partial $(\Delta+1)$-edge coloring to a constant fraction of a {\em batch} of uncolored edges, when they admit a small vertex cover.
For  comparison, as implicit in \cite{BhattacharyaCCSZ24}, the runtime bound for the same task is $$\tilde{O}\!\brac{|E^\star|\Delta + \min_{\tau\geq 1}\left\{\frac{\Delta m|U^\star|\tau}{|E^\star|} + \frac{|E^\star| n}{\tau}\right\}},$$ which is always worse than the new bound we derive in Lemma~\ref{lm:key}; refer to \Cref{sec:parts} (see the remark following \Cref{lem:key 1}) for a detailed discussion on this matter. In addition, our proof of Lemma~\ref{lm:key} is arguably much simpler and more intuitive than the analysis in~\cite{BhattacharyaCCSZ24}; we defer the complete, self-contained proof of this lemma to Section~\ref{sec:star}.

We are now ready to address the concern we pointed out after recurrence~(\ref{eq:recurrence}).   Lemma~\ref{lm:key}, along with the recent breakthrough result of~\cite{Assadi24}, implies Lemma~\ref{lm:main} stated below. This gives us the desired  tradeoff between the runtime of a $(\Delta+1)$-color extension subroutine and the runtime of a $(\Delta+1)$-edge coloring algorithm.  Section~\ref{sec:lm:main} outlines the main idea behind the proof of Lemma~\ref{lm:main}.

\begin{lemma}
\label{lm:main}
Let there be a $(\Delta+1)$-color extension subroutine with $\tilde{O}(\Delta^{\gamma})$ expected runtime, for some $\gamma \geq 1$. Then there is a $(\Delta+1)$-edge coloring algorithm with $\tilde{O}\!\left(m n^{(\gamma -1)/(2\gamma)} \right)$ expected runtime.
\end{lemma}

If we set $\gamma = 1$ in the above lemma, then as expected it leads to a $\tilde{O}(m)$ time $(\Delta+1)$-edge coloring algorithm; and as long as $\gamma < 3$, we get a polynomial improvement over the state-of-the-art $\tilde{O}(mn^{1/3})$ time bound of~\cite{BhattacharyaCCSZ24}. Thus, ideally we would like to obtain a $(\Delta+1)$-color extension subroutine with (say) an expected runtime of $\tilde{O}(\Delta^{2.99})$. But this brings us in front of a significant technical challenge, as explained below.

\medskip
\noindent {\bf A Major Challenge:} In recent years, an influential line of work~\cite{duan2019dynamic,SuV19,ChangHLPU20,GrebikP20,Bernshteyn22,Christiansen23} has addressed the question of finding a {\em small augmenting subgraph} to extend a partial coloring to an uncolored edge, primarily from a different vantage point of distributed algorithms. 
Many of these results were obtained using \emph{multi-step Vizing chains}, a generalization of the central object used in Vizing's original proof \cite{Vizing}.
Couched in our language, this is closely related to  designing a color extension subroutine with small {\em cost}. In particular, for a $(\Delta+1)$-color extension subroutine, the current state-of-the-art bounds are by~\cite{Christiansen23} and~\cite{Bernshteyn22}, who respectively achieve $O(\Delta^7 \log n)$ and $O(\Delta^6 \log^2 n)$ costs using multi-step Vizing chains. It therefore remains an outstanding open question to design such a subroutine with a cost, and moreover with a {\em runtime}, of $\tilde{O}(\Delta^{2.99})$, as demanded by Lemma~\ref{lm:main}. 
We demonstrate that this barrier can be {\em bypassed} by instead achieving the time bound summarized in the theorem below; in fact, proving this theorem is our main technical contribution.

\begin{theorem}
\label{th:extension:main}
    Given a graph $G = (V, E)$ and a partial $(\Delta + 1)$-edge coloring $\chi$ of $G$ with an uncolored edge $e \in E$, we can extend $\chi$ to the edge $e$ in $\tilde O(\Delta^2 + \sqrt{\Delta n})$ expected time.
\end{theorem}

We present a  detailed outline of the proof of Theorem~\ref{th:extension:main} in Section~\ref{sec:th:extension:main}. For now, we focus on emphasizing the following two  conceptual, take-home messages. First, Theorem~\ref{th:extension:main} is sufficient for making a polynomial improvement over the state-of-the-art $\tilde{O}(mn^{1/3})$ runtime for computing a $(\Delta+1)$-edge coloring~\cite{BhattacharyaCCSZ24}. As a simple sanity check, consider the following two cases.
\begin{enumerate}[(i)]
\item $\Delta < n^{1/4}$. Here, we can apply any existing $\tilde{O}(m\Delta)$ time algorithm~\cite{arjomandi1982efficient,gabow1985algorithms,sinnamon2019fast} to compute a $(\Delta+1)$-edge coloring in $\tilde{O}(mn^{1/4})$ time.
\item $\Delta > n^{1/4}$. Here, it is easy to verify that $\Delta^{2.5} \geq \sqrt{\Delta n}$, and hence Theorem~\ref{th:extension:main}  implies a $(\Delta+1)$-color extension subroutine with $\tilde{O}(\Delta^{2.5})$ expected runtime. So we can set $\gamma = 2.5$ in Lemma~\ref{lm:main}, to get a $(\Delta+1)$-edge coloring algorithm with an expected runtime of $\tilde{O}(mn^{3/10})$.
\end{enumerate}
Thus, in both cases we obtain a polynomial improvement over the previous $\tilde{O}(mn^{1/3})$ runtime bound of~\cite{BhattacharyaCCSZ24}. In Section~\ref{sec:full:algo}, we perform a more refined analysis, and derive that Theorem~\ref{th:extension:main} actually implies a $(\Delta+1)$-edge coloring algorithm with $\tilde{O}(mn^{1/4})$ expected runtime (see Theorem~\ref{thm:alg}). This leads us to our main result, as stated in Theorem~\ref{thm:main}.

The second message we wish to emphasize is this: The starting point of our approach for proving Theorem~\ref{th:extension:main} is the multi-step Vizing chain construction of~\cite{duan2019dynamic}.  Specifically,~\cite{duan2019dynamic} designed a $((1+\epsilon)\Delta)$-color extension subroutine with an expected runtime of $\tilde{O}(1/\epsilon^2)$, as long as $\Delta = \tilde{\Omega}(1)$ and $\epsilon = \tilde{\Omega}(1/\sqrt{\Delta})$.
A priori, it seems that such a bound will not be suitable for our purpose, because we want to set $\epsilon = 1/\Delta$. We now outline, at a very high level, the algorithm  of~\cite{duan2019dynamic} and how we  build on top of it. 

\medskip
\noindent {\em In the ensuing paragraphs, we assume that the reader is already familiar with the proof of Vizing's theorem and concepts like ``alternating paths'', ``fans'' and ``Vizing chains''.}

\medskip
\noindent {\bf The~\cite{duan2019dynamic} Algorithm:}
We are given an input graph $G = (V, E)$, and a partial $((1+\epsilon)\Delta)$-edge coloring $\chi$ in $G$ with one uncolored edge $e = (u,v)$. The algorithm of~\cite{duan2019dynamic} proceeds in $R = \Theta(\log n)$ {\em rounds}. For each round $i \in \{1, \ldots, R\}$, it samples a subset of colors $\mathcal{C}_i \subseteq [\Delta+1] \setminus \left(\bigcup_{j < i} \mathcal{C}_j\right)$ of size $|\mathcal{C}_i| = \Theta(\log n/\epsilon)$ independently and u.a.r. We refer to $\mathcal{C}_i$ as the {\em palette} for round $i$. For technical reasons,~\cite{duan2019dynamic} want these palettes across different rounds to be mutually disjoint. In addition, for each $i \in [R]$, they want that  under any fixed partial coloring $\chi'$ of $G$, every vertex $v$ has at least one {\em missing color} in $\mathcal{C}_i$.\footnote{We say that a color is missing at $v$ if it is {\em not} assigned to any of the edges incident on $v$. Since $v$ has degree at most $\Delta$, it has at least $\epsilon \Delta$ missing colors.} To ensure these two properties, it is easy to verify that we must have $|R| \cdot |\mathcal{C}_i| = \Theta(\log^2 n/\epsilon) < \epsilon \Delta$, the RHS being the {\em slack}, in terms of the number of extra available colors. This necessitates the requirement that $\epsilon^2 = \tilde{\Omega}(1/\Delta)$, and hence $\epsilon = \tilde{\Omega}(1/\sqrt{\Delta})$. 

Let $u_1 := u$, $v_1 := v$, $e_1 := (u_1, v_1)$ and $\chi_1 := \chi$. At the start of a given round $i \in [R]$,~\cite{duan2019dynamic} have a partial coloring $\chi_i$ and an uncolored edge $e_i = (u_i, v_i)$ w.r.t.~$\chi_i$. Using only the colors from the palette $\mathcal{C}_i$, they identify a Vizing chain starting from an endpoint (say) $u_i$ of $e_i$. If the alternating path corresponding to this Vizing chain has length less than $L := \tilde{\Theta}(1/\epsilon^2)$, then they {\em apply} the Vizing chain to extend the partial coloring to $e_i$, and their algorithm terminates. Otherwise, they pick some $\alpha_i \in [L]$ independently and u.a.r., and {\em shift} the position of the uncolored edge {\em from} $e_i$ {\em to} the $\alpha_i^{th}$ edge (say) $e_{i+1} = (u_{i+1}, v_{i+1})$ on the concerned alternating path (say) $P_i$. This is done by appropriately shifting the colors on the Vizing fan and the first $\alpha_i$ edges of $P_i$, which leads to a new partial coloring $\chi_{i+1}$. The algorithm then proceeds to round $(i+1)$. 

The choice of the value of $L$ is dictated by the fact that for technical reasons,~\cite{duan2019dynamic} have to ensure that $L = \Omega(|\mathcal{C}_i|^2)$. The main technical contribution of~\cite{duan2019dynamic} is to show that this algorithm terminates within $R$ rounds, w.h.p.\footnote{The in-expectation guarantee follows because if the algorithm fails, then it can always revert back to the trivial implementation of Vizing's proof to extend the partial coloring to one edge in $O(n)$ time.}

\medskip
\noindent {\bf Our Approach:} At a very high-level, we diverge from~\cite{duan2019dynamic} in two major aspects. First, since we need to set $\epsilon = 1/\Delta$, we cannot afford to have these separate palettes $\{\mathcal{C}_i\}_{i}$ across different rounds, and so we get rid of them altogether. Morally, we observe  that all we need is the following {\em key property}: The pair of colors on the concerned alternating path in each round $i \in [R]$ is  disjoint from the ones used in previous rounds. As long as this key property holds,  the~\cite{duan2019dynamic} analysis continues to work just fine. (This assertion has a big caveat associated with it, namely, the complications that arise from the Vizing fans;  see Section~\ref{sec:algo:hurdle:main} for a discussion on those complications.
But we ignore the fans for now.)

The second point of divergence from~\cite{duan2019dynamic} arises because the key property might not hold after a certain number of rounds. To address this issue, our main insight is to increase the value of the parameter $L$ to be $:= \tilde{\Theta}(1/\epsilon^2 + \sqrt{\Delta n}) = \tilde{\Theta}(\Delta^2 + \sqrt{\Delta n})$. With this new increased value of $L$, we derive the following crucial implication (see Lemma~\ref{lem:cont is good:main}). Let $i \in [R]$ be the first round such that: most of the $L$ choices for the value of $\alpha_i$ lead to a scenario where the subsequent round $i+1$ will not satisfy the key property. Then with sufficiently large probability, the algorithm will terminate in the next round (i.e., in round $i+1$). This leads us to a {\em win-win} framework. Either round $i$ is such that with good probability we can continue to apply the analysis of~\cite{duan2019dynamic} in the subsequent round, or with good probability our algorithm actually terminates in the subsequent round. Section~\ref{sec:th:extension:main} contains a much more detailed  exposition of this analysis.

Finally,  we need to overcome further significant obstacles to deal with the Vizing fans. The complete proof of Theorem~\ref{th:extension:main}, which handles these obstacles, is deferred to Section~\ref{sec:extension}.\footnote{The proof in \Cref{sec:extension} is self-contained and uses slightly different notation than the proof sketch in \Cref{sec:th:extension:main}.}

\section{Proof (Sketch) of Lemma~\ref{lm:main}}
\label{sec:lm:main}

By applying Lemma~\ref{lm:key} at most $O(\log n)$ times, where in each of these applications we have at most $|E^*|$ and at least $|L|$ uncolored edges, we derive the following corollary.

\begin{corollary}
\label{cor:key} Suppose that  we receive as input: (i) a partial $(\Delta+1)$-coloring $\chi$ of $G$ as specified in Lemma~\ref{lm:key} and (ii) a parameter $L \in \{1, \ldots, |E^{\star}|\}$. Then in $\tilde{O}(|E^{\star}| \Delta + \Delta m |U^{\star}|/L)$ expected time, we can obtain a partial $(\Delta+1)$-coloring $\chi'' : E \setminus E'' \rightarrow [\Delta+1]$ of $G$, with $E''$ being the set of uncolored edges, such that $E'' \subseteq E^{\star}$ and $|E''| \leq L$.
\end{corollary}

To convey the main idea behind the proof of Lemma~\ref{lm:main},  we will assume that the input graph $G = (V, E)$ is {\em almost} $\Delta$-regular, that is, the degree of each vertex $v \in V$ is $\Omega(\Delta)$ and hence $m = \Theta(n \Delta)$. We also assume that $\Delta = \omega(\log n)$; otherwise, we can  apply an existing algorithm~\cite{arjomandi1982efficient,gabow1985algorithms,sinnamon2019fast} to compute a $(\Delta+1)$-edge coloring of $G$ in $\tilde{O}(m\Delta) = \tilde{O}(m)$ time.

We sample each vertex $v \in V$ into a set $U^{\star} \subseteq V$ independently with probability $(\kappa \log n)/\Delta$, for a sufficiently large constant $\kappa > 1$. Let $E^{\star} := \{ (u, v) \in E : \{u, v\} \cap U^{\star} \neq \emptyset\}$ denote the set of edges incident on $U^{\star}$, and let $E_0 := E \setminus E^{\star}$ denote the set of remaining edges in $G$. Define the subgraphs $G^{\star} := (V, E^{\star})$ and $G_0 := (V, E_0)$. 

Consider any vertex $v \in V$. Its degree in $G$ is at most $\Delta$, and each of its  neighbors gets sampled in $U^{\star}$ with probability $(\kappa \log n)/\Delta$. Thus, applying standard Chernoff bounds, we infer that the degree of $v$ in $G_0$ is at most $(\Delta - \kappa' \log n)$ w.h.p., for a sufficiently large constant $\kappa' > 1$ that depends on $\kappa$. Taking a union bound over all $v \in V$, we get that whp the maximum degree in $G_0$ is at most $(\Delta-\kappa' \log n)$. Based on this observation, we now compute a $(\Delta+1)$-coloring $\chi$ of $G_0$ in $\tilde{O}(m)$ time, by invoking the algorithm of~\cite{Assadi24}. Note that $\chi$ is a partial $(\Delta+1)$-edge coloring of the entire graph $G$, with $E^{\star}$ being the set of uncolored edges and $U^{\star}$ being a vertex cover of $G^{\star} := (V, E^{\star})$, just as in the statement of Lemma~\ref{lm:key}.

At this point, we fix an $L \in \{1, \ldots, |E^{\star}|\}$ whose value will be determined later. Next, we apply Corollary~\ref{cor:key} on $(\chi, G, U^{\star}, E^{\star}, L)$ to obtain a partial coloring $\chi'' : E \setminus E'' \rightarrow [\Delta+1]$, with $E'' \subseteq E$ being the set of remaining uncolored edges and $|E''| \leq L$. It now remains to extend the partial $(\Delta+1)$-coloring $\chi''$ to the edges in $E''$. Towards this end, we fork into one of the following  cases.

\medskip
\noindent {\bf Case (i): $\Delta^{\gamma} > n$.} In this case, we scan through the edges in $E''$. While considering an edge $e \in E''$ during this scan, we extend the current partial coloring  to  $e$,  {\em using a $(\Delta+1)$-color extension subroutine implied by Vizing's original proof which runs in $O(n)$ time}. Thus, overall we spend $O(n  |E''|) = O(n L)$ time to extend the partial $(\Delta+1)$-coloring $\chi''$ to all the edges in $E''$. 

\medskip
\noindent {\bf Case (ii): $\Delta^{\gamma} \leq n$.} Here, the basic set up remains the same as in Case (i) above, except the following: While considering an edge $e \in E''$ during the scan, {\em we apply the $(\Delta+1)$-color extension subroutine  that runs in $\tilde{O}(\Delta^{\gamma})$ expected time}. Thus, overall we spend $\tilde{O}(\Delta^{\gamma}  |E''|) = \tilde{O}(\Delta^{\gamma} L)$ expected time to extend the partial $(\Delta+1)$-coloring $\chi''$ to all the edges in $E''$. 

\medskip
It is easy to see that at the end of the above process, we obtain a $(\Delta+1)$-edge coloring of the entire graph $G = (V, E)$. We now focus on analyzing the runtime of this algorithm. We first recall that each vertex $v \in V$ is sampled into $U^{\star}$ with probability $(\kappa \log n)/\Delta$. Thus, by standard Chernoff bounds we have the following guarantees w.h.p. 
\begin{equation}
\label{eq:sketch:1}
|U^{\star}| = \tilde{O}(n/\Delta), \text{ and hence } |E^{\star}| \leq |U^{\star}| \Delta = \tilde{O}(n).
\end{equation}

We have already observed that it takes $\tilde{O}(m)$ time to compute the partial coloring $\chi$. Given $\chi$, Corollary~\ref{cor:key} allows us to obtain the partial coloring $\chi''$ in expected time $\tilde{O}(|E^{\star}| \Delta + \Delta m |U^{\star}|/L) = \tilde{O}(n \Delta + mn/L) = \tilde{O}(m + n^2 \Delta/L)$. The first equality holds because of~(\ref{eq:sketch:1}), and the second inequality holds because we have assumed that the input graph is almost regular. Accordingly, the overall runtime of the algorithm depends on whether we are in Case (i) or Case (ii), as follows.

If we are in Case (i), then the  total expected runtime of the algorithm is given by $T := \tilde{O}(m) + \tilde{O}(m+n^2 \Delta/L) + O(nL)$. Setting $L = \sqrt{n \Delta}$, we get  
$$T = \tilde{O}(m+ n \sqrt{n \Delta}) = \tilde{O}(m + n \Delta \sqrt{n/\Delta}) = \tilde{O}(m \sqrt{n/\Delta}) = \tilde{O}\left(m n^{(\gamma-1)/(2\gamma)}\right),$$ where the last inequality holds because $\Delta^{\gamma}> n$ (which implies that $1/\Delta < 1/n^{1/\gamma}$). 

In contrast, if we are in Case (ii), then the total expected runtime of the algorithm is given by $T := \tilde{O}(m) + \tilde{O}(m+n^2 \Delta/L) + \tilde{O}(\Delta^{\gamma} L)$. Setting $L = n/\Delta^{(\gamma-1)/2}$, we get  
$$T =  \tilde{O}\left(m + n \Delta^{(\gamma+1)/2}\right) = \tilde{O}\left(m + n\Delta \cdot  \Delta^{(\gamma-1)/2}\right) = \tilde{O}\left(m   \Delta^{(\gamma-1)/2}\right)  = \tilde{O}\left(m n^{(\gamma-1)/(2\gamma)}\right).$$ where the last inequality holds because $\Delta^{\gamma} \leq n$. This concludes the proof (sketch) of Lemma~\ref{lm:main}.

\section{Proof (Sketch) of Theorem~\ref{th:extension:main}}
\label{sec:th:extension:main}

We define the parameters:
\begin{equation}
\label{eq:parameters:main}
\ell := 100 \log n, \ \ L := 10^3 \ell^2 (\Delta^2 + \sqrt{\Delta n})
\end{equation}
To convey the main conceptual ideas behind our analysis, in this section we will explain the proof of Theorem~\ref{th:extension:main} under  Assumption~\ref{assume:main}, stated below. This assumption  clearly holds, for example, on bipartite graphs. For those readers familiar with the proof of Vizing's theorem, this assumption will allow us to ignore the Vizing fans altogether, and we will be able to focus only on the alternating paths, which are more intuitive to reason about. 

\begin{assumption}
\label{assume:main} The graph $G = (V, E)$ does not contain any odd cycle of length $\leq 2L+3$.
\end{assumption}

In Section~\ref{sec:algo:prelim:main}, we introduce some key notations and terminologies regarding alternating paths. Next, we describe our randomized algorithm in Section~\ref{sec:algo:main}. In Section~\ref{sec:algo:metatree:main}, we introduce a recursion tree (which we refer to as the ``meta-tree'') which encodes all possible execution-paths of our algorithm,
and present a few basic properties of this meta-tree. We show that our algorithm essentially performs a random walk on this meta-tree, and state Lemma~\ref{lm:extension:main} which guarantees that this random walk terminates quickly. We prove Lemma~\ref{lm:extension:main}, which implies Theorem~\ref{th:extension:main}, in Section~\ref{sec:algo:analysis:main}. Finally, we conclude our discussion in Section~\ref{sec:algo:hurdle:main} by pointing out the remaining significant technical challenges that we need to overcome, if we are to get rid of Assumption~\ref{assume:main}.

\subsection{Preliminaries}
\label{sec:algo:prelim:main}

\noindent {\bf Alternating Paths and Types:} Consider a partial $(\Delta+1)$-edge coloring $\chi : E \rightarrow [\Delta+1] \cup \{ \bot \}$ in the input graph $G = (V, E)$. Consider a path $P = ( (v_0, v_1), (v_1, v_2), \ldots, (v_{k-1}, v_k))$ in $G$, where $(v_{i-1}, v_i) \in E$ is the $i^{th}$ edge on the path, for all $i \in [k]$. We say that $P$ is an {\bf alternating path}  (w.r.t.~$\chi$) iff there exist two distinct colors $c, c' \in [\Delta+1]$  that satisfy the following conditions.
\begin{enumerate}
\item The colors on the consecutive edges of $P$ alternate between $c$ and $c'$. Specifically, this means that $\chi(v_{i-1}, v_i) \in \{c, c'\}$ for all $i \in [k]$, and $\chi(v_{i-1}, v_{i}) \neq \chi(v_i, v_{i+1})$ for all $i \in [k-1]$.
\item The path $P$ is {\em maximal}. Thus,   for each vertex $u \in \{v_0, v_k\}$,  either $c \in \texttt{miss}_{\chi}(u)$ or $c' \in \texttt{miss}_{\chi}(u)$, where $\texttt{miss}_{\chi}(u) \subseteq [\Delta+1]$ denotes the set of {\em missing colors} at $u$ under $\chi$ (i.e., these are the colors that are {\em not} assigned to any edge in $G$ incident on $u$). 
\end{enumerate}
Let $\tau = \{c, c'\}$. We refer to $\tau$ as being the {\bf type} of the alternating path $P$. We let $\texttt{length}(P) = k$ denote the length of the path $P$.  We also say that the path $P$ {\bf starts} at $v_0$ and {\bf ends} at $v_k$. For  $i \in [k]$, we let $P_{\leq i} = ((v_0, v_1), (v_1, v_2), \ldots, (v_{i-1}, v_i))$ denote the {\bf length-$i$ prefix} of $P$. For $i > k$, we let  $P_{\leq i} := P$. Finally, note that an alternating path always comes with an associated  {\bf orientation}. In particular, there is another alternating path with the same set of edges as $P$, but in reverse order.

\medskip
Our algorithm in Section~\ref{sec:algo:main} will use two basic subroutines, as described below.

\medskip
\noindent {\bf The Subroutine $\textsf{Apply}(\chi, e, P)$:} Here, the input is a partial $(\Delta+1)$-edge coloring $\chi$ of $G$, an uncolored edge $e = (u, v)$ and an alternating path $P$ starting from $u$ that is of type $\tau = \{\alpha_u, \alpha_v\}$, where $\alpha_u \in \texttt{miss}_{\chi}(u) \setminus \texttt{miss}_{\chi}(v)$ and $\alpha_v \in \texttt{miss}_{\chi}(v) \setminus \texttt{miss}_{\chi}(u)$.
(If either $\alpha_u \in \miss_\chi(u) \cap \miss_\chi(v)$ or 
$\alpha_v \in \miss_\chi(u) \cap \miss_\chi(v)$,  
our algorithm 
will not make the call $\textsf{Apply}(\chi,e,P)$.)
Crucially, it is  guaranteed that $\texttt{length}(P) \leq 2L+2$. W.l.o.g., suppose that  $P = ((v_0, v_1), (v_1, v_2), \ldots, (v_{k-1}, v_k))$, where $k = \texttt{length}(P)$ and $v_0 = u$, and let $c_k = \chi(v_{k-1}, v_k) \in \{\alpha_u, \alpha_v\}$.  It is easy to verify
that under Assumption~\ref{assume:main}, we have $v_k \neq v$,  for otherwise the path $P$ along with the edge $(u, v)$ would create an odd cycle of length  $\leq 2L+3$. The subroutine updates the coloring $\chi$ as follows.
\begin{itemize}
\item $\chi(v_{i-1}, v_i) \leftarrow \chi(v_i, v_{i+1})$ for all $i \in [k-1]$.
\item $\chi(v_{k-1}, v_k) \leftarrow c$, where  $c$ is the unique color in $\{\alpha_u, \alpha_v\} \setminus \{c_k\}$.
\item $\chi(u, v) \leftarrow \alpha_v$.
\end{itemize}
At the end of the above operations, the partial coloring $\chi$ gets extended to the edge $e$. The subroutine returns the updated partial coloring. In summary, the subroutine {\em applies} the alternating path $P$ to extend the partial coloring to $e$.

\medskip
\noindent {\bf The Subroutine $\textsf{Shift}(\chi, e, P_{\leq i})$:} Here, the input is a partial $(\Delta+1)$-edge coloring $\chi$ of $G$, an uncolored edge $e = (u, v)$ and a length-$i$  prefix $P_{\leq i}$ of an alternating path $P$ starting from $u$ that is of type $\tau = \{\alpha_u, \alpha_v\}$, where $\alpha_u \in \texttt{miss}_{\chi}(u) \setminus \texttt{miss}_{\chi}(v)$ and $\alpha_v \in \texttt{miss}_{\chi}(v) \setminus \texttt{miss}_{\chi}(u)$.  
(If either $\alpha_u \in \miss_\chi(u) \cap \miss_\chi(v)$ or 
$\alpha_v \in \miss_\chi(u) \cap \miss_\chi(v)$,  
our algorithm 
will not make the call $\textsf{Shift}(\chi, e, P_{\leq i})$.)
Crucially, we are also guaranteed that $i \leq 2L+2$. W.l.o.g., let $P_{\leq i} = ((v_0, v_1), (v_1, v_2), \ldots, (v_{i-1}, v_i))$, where $v_0 = u$. Due to Assumption~\ref{assume:main}, it is again easy to verify that $v_i \neq v$.  The subroutine updates  $\chi$ as follows.
\begin{itemize}
\item $\chi(v_{j-1}, v_j) \leftarrow \chi(v_{j}, v_{j+1})$ for all $j \in [i-1]$.
\item $\chi(v_{i-1}, v_i) \leftarrow \bot$.
\item $\chi(u, v) \leftarrow \alpha_v$.
\end{itemize}
At the end of the above operations, in the partial coloring $\chi$ the position of the uncolored edge gets {\em shifted} from $e$ to $(v_{i-1}, v_i)$. The subroutine returns the updated coloring.

\subsection{Our Algorithm}
\label{sec:algo:main}

We have a graph $G = (V, E)$, and a partial $(\Delta+1)$-edge coloring $\chi$ of $G$ with one uncolored edge $e = (u,v)$. We need to extend $\chi$ to  $e$. We do this by  identifying two colors $c_u \in [\Delta+1] \setminus \miss_{\chi}(u)$ and $c_v \in [\Delta+1] \setminus \miss_{\chi}(v)$, and then  calling Algorithm~\ref{alg:extend:main} with input $(G, \chi, e = (u,v), \{c_u, c_v\})$.

While parsing the pseudocode of Algorithm~\ref{alg:extend:main}, the reader should think of $c_u$ and $c_v$ as {\bf blocking colors}. The algorithm considers an  alternating path $P$ starting from $u$, such that the type of this alternating path, given by $\{\alpha_u, \alpha_v\}$, is disjoint from $\{c_u, c_v\}$. Note that it is always possible to find such a type, for the following reason. Since  $u$ has degree at most $\Delta$ in $G$ and the edge $e = (u, v)$ is currently uncolored, we have $|\texttt{miss}_{\chi}(u)| \geq (\Delta+1) - (\Delta-1) = 2$. As $c_u \notin \texttt{miss}_{\chi}(u)$, it must be the case that $\texttt{miss}_{\chi}(u) \setminus \{c_u, c_v\} \neq \emptyset$, and so there exists an appropriate color $\alpha_u$ that we can pick from  $\texttt{miss}_{\chi}(u) \setminus \{c_u, c_v\}$. A similar argument holds for $\alpha_v$. Also, by induction, it is easy to verify that all future recursive calls to the algorithm will continue to satisfy the same property with regard to the blocking colors. Specifically, if the algorithm is called with input $(G, \chi,  (u', v'), \{c_1, c_2\})$, then $\texttt{miss}_{\chi}(w) \cap \{c_1, c_2\} \neq \emptyset$ for each $w \in \{u', v'\}$. The reason for having these blocking colors will become apparent later on (see Observation~\ref{ob:parent:child:main} and the proof of Lemma~\ref{lem:cont is good:main}).

\medskip

\begin{algorithm}[H]
    \SetAlgoLined
    \DontPrintSemicolon
    \SetKwRepeat{Do}{do}{while}
    \SetKwBlock{Loop}{repeat}{EndLoop}
    Find two colors $\alpha_u \in \miss_{\chi}(u) \setminus \{c_u, c_v\}$ and $\alpha_v \in \miss_{\chi}(v) \setminus \{c_u, c_v\}$\label{line:start}\;
    \If{$\exists c \in \{\alpha_u, \alpha_v\}$ such that $c \in \miss_{\chi}(u) \cap \miss_{\chi}(v)$}
    {
    $\chi(u, v) \leftarrow c$ \;
    \Return{$\chi$}\label{line:end 1}
    }
    Let $P := (e_1, \ldots, e_k)$ be the $\{\alpha_u, \alpha_v\}$-alternating path in $G$ (w.r.t.~$\chi$) starting from $u$\label{line:5}\;
    \If{$k \leq 2L+2$}
    {$\chi \leftarrow \textsf{Apply}(\chi, e , P)$ \;
    \Return{$\chi$}\label{line:end 2}
    }
    \Else
    {
    Sample $i \in \{1, \ldots, L\}$ independently and u.a.r.\label{line:choices}\;
    Let $P_{\leq i} = (e_1, \ldots, e_i)$ denote the prefix of $P$ consisting of its first $i$ edges\;
    $\chi \leftarrow \textsf{Shift}(\chi, e, P_{\leq i})$\;
    $\textsf{ExtendColoring}(G, \chi, e_i, \{\alpha_u, \alpha_v\})$
    }
    \label{line:end}
    \caption{\textsf{ExtendColoring}$(G,\chi, e = (u, v), \{c_u, c_v\})$}
    \label{alg:extend:main}
\end{algorithm}

\medskip

\noindent
The remainder of Algorithm~\ref{alg:extend:main} is very natural and intuitive. If there exists some color $c \in \{\alpha_u, \alpha_v\}$ that is missing {\em both} at $u$ and $v$, then  the algorithm simply assigns the color $c$ to $e$, and terminates. Otherwise, if the length of the alternating path $P$ is at most $2L+2$, then it extends the partial coloring to $e$ by applying the alternating path, and terminates. Finally, if $\texttt{length}(P) > 2L+2$, then it picks some $i \in [L]$  u.a.r., shifts the position of the uncolored edge from $e$ to the $i^{th}$ edge on $P$, and makes a recursive call to itself.

\subsection{The Meta-Tree}
\label{sec:algo:metatree:main}

We now define a  (possibly infinite) tree $\mathcal{T}$ that captures all possible execution paths taken by our recursive algorithm, on a given input. To clearly distinguish it from the input graph $G$, we refer to $\mathcal{T}$ as a {\bf meta-tree} and its vertices as {\bf meta-nodes}. We next introduce some key notations.

The meta-tree $\mathcal T$ is rooted at a meta-node $r$. We let $V(\mathcal{T})$ denote the set of all meta-nodes in $\mathcal{T}$. 
Every meta-node $x \in V(\mathcal T)$ corresponds to a recursive {\bf call} of \Cref{alg:extend:main}, and the root-to-leaf path from $r$ to $x$ in $\mathcal T$ corresponds to the execution of our recursive algorithm leading to this call.\footnote{While giving the full proof in \Cref{sec:extension}, we formulate our algorithm iteratively instead of recursively and refer to \textbf{iterations} instead of \textbf{calls}.}
For each such meta-node $x$, we denote the state of an entity $\Phi$  at the start of the corresponding call of Algorithm~\ref{alg:extend:main} by $\Phi^{(x)}$. For example, we use the symbols $\chi^{(x)}$, $P^{(x)}$ and $\tau^{(x)}$ respectively to denote the following entities at the start of the concerned call of Algorithm~\ref{alg:extend:main}: (i) the current partial coloring $\chi$, (ii) the alternating path $P$ (see \Cref{line:5} of Algorithm~\ref{alg:extend:main}) and (iii) the type $\{\alpha_u, \alpha_v\}$ of the path $P$ (see \Cref{line:5} of Algorithm~\ref{alg:extend:main}). 

If a meta-node $x$ is {\em not} a leaf in $\mathcal{T}$, then it has exactly $L$ children, one for each choice of $i \in [L]$ in \Cref{line:choices} of Algorithm~\ref{alg:extend:main}. In contrast, a meta-node $x$ is a leaf in $\mathcal{T}$ iff the corresponding call of Algorithm~\ref{alg:extend:main} terminates at either \Cref{line:end 1} or \Cref{line:end 2}. We will refer to the leaves in $\mathcal{T}$ as {\bf terminals}.

\medskip
\noindent \textbf{Random Walks in $\mathcal T$:} 
An execution of our algorithm defines a {\bf random walk} on $\mathcal T$ that starts at the root and, at each step, independently and uniformly samples a random child of the current meta-node. The random walk ends when, and if, it reaches a terminal meta-node: At that point the algorithm succeeds in extending the initial partial coloring $\chi$ to the uncolored edge given to it as part of the input. Using standard data structures, it is easy to ensure that the algorithm spends $\tilde{O}(L)$ time at each meta-node it visits during this random walk, because the colors of at most $O(L)$ edges get changed during any given call of Algorithm~\ref{alg:extend:main}. Since $L = \tilde{O}(\Delta^2 + \sqrt{\Delta n})$, Theorem~\ref{th:extension:main}  follows from Lemma~\ref{lm:extension:main} and standard tools for boosting the success probability of a randomized algorithm. We derive some basic properties of the meta-tree in Section~\ref{subsec:prop}. Subsequently, we use these properties to prove Lemma~\ref{lm:extension:main} in Section~\ref{sec:algo:analysis:main}.

\begin{lemma}
\label{lm:extension:main}
With probability $\Omega(1/\log n)$, the random walk on $\mathcal{T}$ executed by our algorithm ends at a terminal vertex that lies within  depth $O(\log n)$ from the root $r$.
\end{lemma}

\subsubsection{Basic Properties of the Meta-Tree}
\label{subsec:prop}

 Recall that $\tau^{(x)}$ denotes the type of the alternating path $P^{(x)}$ at a meta-node $x$.
 Consider a non-terminal meta-node $x$ at {\bf depth} $= k$ (say), and suppose that $(x_0, x_1, \ldots, x_k)$ is the unique path in $\mathcal{T}$ from the root to $x$ (i.e., $r = x_0$ and $x = x_k$). We say that the meta-node $x$ is {\bf dirty} iff $\tau^{(x)} \cap \left( \tau^{(x_0)} \cup \cdots \cup \tau^{(x_{k-1})} \right) \neq \emptyset$, and {\bf clean} otherwise. We also say that $x$ is {\bf contaminated} iff at least $L/(10 \ell)$ of its children are dirty. Note that a contaminated meta-node itself might be clean.  

 \medskip
 \noindent {\bf Remark:} We emphasize that according to our definitions  {\em only} the non-leaf meta-nodes are classified as being either clean or dirty. Thus, the set of meta-nodes is partitioned into three subsets: terminal, dirty and clean. Furthermore, a subset of non-terminal meta-nodes are contaminated. This implies that some of the contaminated meta-nodes are clean, the rest being dirty.

 \medskip
 We next derive a few key properties that will be  useful in proving Lemma~\ref{lm:extension:main} later on.

\begin{observation}
\label{ob:parent:child:main}
Consider any non-terminal meta-node $x \in V(\mathcal{T})$, and let $y$ be any child of $x$. If $y$ is {\em not} a terminal, then we must have $\tau^{(x)} \cap \tau^{(y)} = \emptyset$.
\end{observation}

\begin{proof}
Follows from \Cref{line:start} and \Cref{line:end} of Algorithm~\ref{alg:extend:main}.
\end{proof}

\begin{lemma}\label{lem:cont is good:main}
    Consider any contaminated meta-node $x \in V(\mathcal{T})$ at a depth $\leq \ell$ in the meta-tree $\mathcal T$. 
    Then at least  $L/(40\ell)$  many children of $x$ are terminal meta-nodes.
\end{lemma}

\begin{proof}
    Let $(x_0, x_1, \ldots, x_k)$ denote the unique path from the root to $x$ in $\mathcal{T}$, with $r = x_0$ and $x = x_k$. Thus, at most $2(k+1) \leq 2 (\ell+1) \leq 4\ell$ distinct colors appear in the set $\tau^{(x_0)} \cup \tau^{(x_1)} \cup \cdots \cup \tau^{(x_k)} = \mathcal{C}^{\star}$ (say); because
    $|\tau^{(x_i)}| = 2$ for all $i \in [k]$. Let $K^{\star} := \{ \tau \in  \binom{[\Delta+1]}{2} : \tau \cap \mathcal{C}^{\star} \neq \emptyset\}$ denote the collection of types with at least one color from $\mathcal{C}^{\star}$. Note that $|K^{\star}| \leq |\mathcal{C}^{\star}| (\Delta+1)  \leq 4\ell(\Delta+1) \leq 8 \ell \Delta$. Let $D$ denote the set of dirty children of $x$. By definition, for each meta-node $y \in D$, we have $\tau^{(y)} \in K^{\star}$. Furthermore, since $x$ is contaminated, it follows that $|D| \geq L/(10\ell)$. 

    Say that a meta-node is {\bf trivial} iff the corresponding call of Algorithm~\ref{alg:extend:main} ends at \Cref{line:end 1}. Thus, every trivial meta-node is terminal, but not vice versa. Let $D_t \subseteq D$ denote the set of trivial meta-nodes that belong to  $D$. If at least half of the meta-nodes in $D$ are trivial, then  $|D_t| \geq |D|/2 \geq L/(20\ell)$, and the lemma follows. Thus, for the rest  of the proof we assume that:
    \begin{equation}
    \label{eq:trivial}
    |D \setminus D_t| \geq |D|/2 \geq L/(20 \ell).
    \end{equation}
    Since $\tau^{(y)} \cap \tau^{(x)} = \emptyset$ for all $y \in D \setminus D_t$ (see Observation~\ref{ob:parent:child:main}), all the  alternating paths in the collection $\mathcal{P}' := \{ P^{(y)} \}_{y \in D \setminus D_t}$ exist {\em simultaneously} in the graph $G$ w.r.t.~the partial $(\Delta+1)$-edge coloring $\chi^{(x)}$. To be more precise, this means that for every meta-node $y \in D \setminus D_t$  and every edge $e \in P^{(y)}$, we have $\chi^{(y)}(e) = \chi^{(x)}(e)$.  Furthermore, we have already inferred that each path $P^{(y)} \in \mathcal{P}'$ is of a type $\tau^{(y)} \in K^{\star}$. Next, note that the total length of all possible alternating paths of a given type (w.r.t.~a specific partial coloring) is at most $2n$. This holds because such paths are vertex-disjoint, except the same path being possibly counted twice from two opposite directions. 
    Thus, we have $
    \sum_{P \in \mathcal{P}'} \texttt{length}(P) \leq 2n |K^{\star}|  \leq 16 \ell \Delta n$. Hence, the average length of an alternating path $P^{(y)} \in \mathcal{P}'$ is at most $16\ell  \Delta  n/|\mathcal{P}'| = 16 \ell \Delta n/|D \setminus D_t| \leq 320 \ell^2 \Delta n/L \leq L$, where the second-last inequality follows from~(\ref{eq:trivial}) and the last inequality\footnote{This is the only place in our analysis where we require $L$ to be  $\tilde{\Omega}(\sqrt{\Delta n})$.} follows from~(\ref{eq:parameters:main}).   This implies that at least half of the alternating paths in $\mathcal{P}'$ have length at most $2L$. So, at least half of the meta-nodes $y \in D \setminus D_t$ have $\texttt{length}\brac{P^{(y)}} \leq 2L$; and such meta-nodes are terminals. We therefore conclude that at least $|D \setminus D_t|/2 \geq L/(40 \ell)$ children of $x$ are terminal meta-nodes. 
\end{proof}

We say that a meta-node $x \in V(\mathcal{T})$ is {\bf congenitally clean} iff  every ancestor of $x$, along with $x$ itself, is clean. Next, consider any meta-node $x \in V(\mathcal{T})$ at depth (say) $k$ in $\mathcal{T}$. Let $(x_0, x_1, \ldots, x_k)$ denote the unique path from the root to $x$ in $\mathcal{T}$, with $r = x_0$ and $x = x_k$. Then we refer to the ordered tuple of types $\left(\tau^{(x_0)}, \tau^{(x_1)}, \ldots, \tau^{(x_k)}\right)$ as the {\bf transcript} of $x$. We now upper bound the number of congenitally clean vertices with the same transcript.

\begin{lemma}\label{lem:expansion:main}
 Let $\tau_0,\dots,\tau_i$ be a sequence of types. Then there can be at most $n$ congenitally clean meta-nodes with transcript $= (\tau_0, \ldots, \tau_i)$.
\end{lemma}

\begin{proof}
For $j \in [0, i]$, let $\Gamma_j \subseteq V(\mathcal{T})$ denote the set of congenitally clean meta-nodes with transcript $= (\tau_0, \ldots, \tau_j)$. Note that every meta-node in $\Gamma_j$ is at depth $= j$ in $\mathcal{T}$. 
\begin{claim}
\label{cl:expansion:main}
For all $j \in [0, i]$, the collection $\left\{P^{(x)}_{\leq L}\right\}_{x \in \Gamma_j}$ of length-$L$ prefixes are  vertex-disjoint.
\end{claim}

Setting $j = i$ in Claim~\ref{cl:expansion:main}, it follows that the size of the set $\Gamma_i$ is at most the maximum possible number of vertex-disjoint paths in the input graph $G$, which in turn, is at most $n$. This implies the lemma. Accordingly, from now on we focus on proving Claim~\ref{cl:expansion:main}. 

We will prove Claim~\ref{cl:expansion:main} via induction on $j$. Since $\Gamma_0 = \{r\}$, the claim trivially holds if $j = 0$. By induction hypothesis, we now assume that there exists an index $j^{\star} \in [0, i-1]$ such that the claim holds for all $j \leq j^{\star}$. Under this assumption, we will show that the claim holds for $j = j^{\star}+1$.

If there is a color that appears more than once across  the types $\tau_0, \ldots, \tau_{j^{\star}+1}$, then  a meta-node with the transcript $(\tau_0, \ldots, \tau_{j^{\star}+1})$ is not congenitally clean, and so $\Gamma_{j} = \emptyset$ for all $j \in [j^{\star}+1, i]$. Henceforth, we assume that the types $\tau_0, \ldots, \tau_{j^{\star}+1}$ are mutually disjoint.

Consider any two distinct meta-nodes $x,y \in \Gamma_{j^{\star}+1}$. Let $u^{(x)}$ and $u^{(y)}$ respectively denote the starting points of the alternating paths $P^{(x)}$ and $P^{(y)}$. Our induction hypothesis implies that $u^{(x)} \neq u^{(y)}$. Since the types $\tau_0, \ldots, \tau_{j^{\star}+1}$ are mutually disjoint, both the alternating paths $P^{(x)}$ and $P^{(y)}$ exist {\em simultaneously} in $G$ w.r.t.~the initial partial coloring  $\chi^{(r)}$. In other words, either they are two completely  
disjoint type-$\tau_{j^{\star}+1}$ alternating paths in $G$ w.r.t.~$\chi^{(r)}$, or essentially the same path (with the same set of edges) but with a different orientation. In the first case, their length-$L$ prefixes $P^{(x)}_{\leq L}$ and  $P^{(y)}_{\leq L}$ are clearly vertex-disjoint. In the second case, we note that $x$ and $y$ are {\em not} terminal meta-nodes (because every meta-node in $\Gamma_{j^{\star}+1}$ is congenitally clean, by definition); hence both $P^{(x)}$ and $P^{(y)}$ have length  $\geq 2L+2$, and so $P^{(x)}_{\leq L}$ and $P^{(y)}_{\leq L}$ are also vertex-disjoint. This concludes the proof of the claim.
\end{proof}

\subsection{Analyzing the Random Walk on the Meta-Tree: Proof of Lemma~\ref{lm:extension:main}}
\label{sec:algo:analysis:main}

Our analysis will crucially rely on the behavior of the random walk within a certain {\bf critical subtree} $\mathcal{T}^{\star}$ of $\mathcal{T}$. We define this critical subtree below, and then summarize a few of its key properties.

\begin{wrapper}
The ``critical subtree'' $\mathcal{T}^{\star}$ is obtained by starting with $\mathcal{T}$, and then deleting every meta-node $x \in V(\mathcal{T})$ that satisfies at least one of the following conditions: (i) $x$ has a dirty ancestor in $\mathcal{T}$, (ii) $x$ has a contaminated ancestor in $\mathcal{T}$, and (iii) $x$ is at a depth strictly greater than $\ell$ in $\mathcal{T}$. We let $V(\mathcal{T}^{\star}) \subseteq V(\mathcal{T})$ denote the set of meta-nodes in the critical subtree.
\end{wrapper}
Note that if the root $r$ is itself a terminal, then Lemma~\ref{lm:extension:main} trivially holds because the random walk ends at $r$. Thus, for the rest of the proof we assume that {\bf the root $r$ is {\em not} a terminal}.

\begin{observation}
\label{obs:connect:0}
\label{obs:connect}
$\mathcal{T}^{\star}$ is a connected subtree of $\mathcal{T}$, rooted at $r$, and $r$ is {\em not} a leaf in $\mathcal{T}^{\star}$. Also, for every   $x \in V(\mathcal{T}) \setminus V(\mathcal{T}^{\star})$, the unique path in $\mathcal{T}$ from  $r$ to  $x$ passes through some leaf in $\mathcal{T}^{\star}$.
\end{observation}

\begin{proof}
Since the $r$ is at depth $0$ and does not have any ancestor, it belongs to $\mathcal{T}^{\star}$. Furthermore, we have assumed that $r$ is {\em not} a terminal. This implies that, by definition, $r$ is a congenitally clean meta-node. So all the children of $r$ are part of $\mathcal{T}^{\star}$, and hence $r$ is {\em not} a leaf in $\mathcal{T}^{\star}$.

Next, note that if a meta-node $x$ is in $\mathcal{T}^{\star}$, then each of its siblings and each of its ancestors is also in $\mathcal{T}^{\star}$. This implies the observation.
\end{proof}

\begin{observation}
\label{obs:internal}
Every non-leaf meta-node in $\mathcal{T}^{\star}$ is congenitally clean and not contaminated.
\end{observation}

\begin{proof}
Consider any meta-node $x \in V(\mathcal{T}^{\star})$. If $x$ is terminal, then it is a leaf in $\mathcal{T}$ itself, and hence also a leaf in $\mathcal{T}^{\star}$. In contrast, if $x$ is either contaminated or dirty, then it cannot have any descendant in $\mathcal{T}^{\star}$. Thus, for $x$ to be a non-leaf meta-node in $\mathcal{T}^{\star}$, it must be clean and not contaminated. Since we can infer the same for every ancestor of $x$, the observation follows.
\end{proof}

The above  observations help us gain an intuitive understanding of the critical subtree $\mathcal{T}^{\star}$. Define the {\bf core} (resp.~{\bf boundary}) of $\mathcal{T}^{\star}$ to be the set of its non-leaf (resp.~leaf) meta-nodes. Every meta-node within the core is congenitally clean and not contaminated (see Observation~\ref{obs:internal}). The random walk on $\mathcal{T}$ undertaken by our algorithm starts at the root $r$, which is part of the core of $\mathcal{T}^{\star}$ (see Observation~\ref{obs:connect:0}). For a certain number of steps the random walk stays within the core. After that, at some point in time the random walk exits the core by reaching a meta-node (say) $x$ at the boundary of $\mathcal{T}^{\star}$ (see Observation~\ref{obs:connect}). If $x$ is terminal, then  the random walk ends at $x$. Otherwise,  it ventures out of $V(\mathcal{T}^{\star})$ in the subsequent step, and  never comes back to $V(\mathcal{T}^{\star})$ in future.

The above discussion also implies that for every meta-node $x \in V(\mathcal{T}^{\star})$, its depth in $\mathcal{T}^{\star}$ is the same as its depth in $\mathcal{T}$. Accordingly, from this point onward we will use the phrase ``the depth of a meta-node'' without explicitly referring to the underlying meta-tree.

\begin{observation}
\label{obs:partition}
Let $\mathcal{Z}^{\star}$ denote the set of  leaves in $\mathcal{T}^{\star}$. Then $\mathcal{Z}^{\star}$ is partitioned into four subsets:
\begin{itemize}
\item $\mathcal{Z}^{\star}_{\textnormal{\texttt{t}}} := \{ x \in \mathcal{Z}^{\star} : x \textnormal{ is terminal}\}$.
\item $\mathcal{Z}^{\star}_{\textnormal{\texttt{cc-cont}}} := \{ x \in \mathcal{Z}^{\star} : x \textnormal{ is congenitally clean and contaminated}\}$.
\item $\mathcal{Z}^{\star}_{\textnormal{\texttt{cc-not-cont}}} := \{ x \in \mathcal{Z}^{\star} : x \textnormal{ is congenitally clean and {\em not} contaminated}\}$.
\item $\mathcal{Z}^{\star}_{\textnormal{\texttt{d}}} := \{ x \in \mathcal{Z}^{\star} : x \textnormal{ is dirty}\}$.
\end{itemize}
\end{observation}

\begin{proof}
This holds because the set of meta-nodes is partitioned into three substes: terminal, clean and dirty. Finally, every clean meta-node in $\mathcal{T}^{\star}$ is  congenitally clean, because it cannot have any dirty ancestor.
\end{proof}

\begin{corollary}
\label{cor:partition}
Let $\mathcal{E}^{\star}_{\texttt{d}}, \mathcal{E}^{\star}_{\texttt{t}},\mathcal{E}^{\star}_{\texttt{cc-cont}}$ and $\mathcal{E}^{\star}_{\texttt{cc-not-cont}}$ respectively denote the  events that the random walk undertaken by our  algorithm reaches a meta-node in $\mathcal{Z}^{\star}_{\texttt{d}}, \mathcal{Z}^{\star}_{\texttt{t}},\mathcal{Z}^{\star}_{\texttt{cc-cont}}$ and $\mathcal{Z}^{\star}_{\texttt{cc-not-cont}}$. These four events are mutually exclusive and exhaustive, and hence:
$$\Pr\left[\mathcal{E}^{\star}_{\texttt{d}}\right] + \Pr\left[\mathcal{E}^{\star}_{\texttt{t}}\right]
+
\Pr\left[\mathcal{E}^{\star}_{\texttt{cc-cont}}\right]  + \Pr\left[\mathcal{E}^{\star}_{\texttt{cc-not-cont}}\right] = 1.$$
\end{corollary}

\begin{proof}
Follows from Observation~\ref{obs:connect} and Observation~\ref{obs:partition}.
\end{proof}

\begin{observation}
\label{obs:depth}
Every meta-node in $\mathcal{Z}^{\star}_{\texttt{cc-not-cont}}$ is at depth $\ell$.
\end{observation}

\begin{proof}
Consider any meta-node $x \in \mathcal{Z}^{\star}_{\texttt{cc-not-cont}}$. By definition, the depth of $x$ is no more than $\ell$. Suppose that $x$ is at a depth (say) $k < \ell$. Since $x$ is clean, it has  $L$ children in the meta-tree $\mathcal{T}$. 

Let $y$ be any  child of $x$ in $\mathcal{T}$. Since $x$ is congenitally clean, $y$ does not have any dirty ancestor. Also, since $x$  is {\em not contaminated}, $y$ cannot have any contaminated ancestor; for otherwise $x$ itself would have had the same contaminated ancestor and so $x$ would not be part of $\mathcal{T}^{\star}$. Finally, the meta-node $y$ is at depth $= k+1 \leq \ell$.  We therefore conclude that $y$ is part of $\mathcal{T}^{\star}$. But this contradicts our  assumption that $x$ is a leaf in $\mathcal{T}^{\star}$. So, the meta-node $x$ must be at depth $k = \ell$.
\end{proof}

Armed with these basic observations about the critical subtree, we are now ready to prove Lemma~\ref{lm:extension:main}. Our strategy will be to show that the event $\mathcal{E}^{\star}_{\texttt{cc-cont}} \cup \mathcal{E}^{\star}_t$ 
occurs with constant probability, and conditioned on this event, the random walk ends within $O(\log n)$ steps with probability $\Omega(1/\log n)$.

\begin{claim}
\label{cl:clean}
We have $\Pr\left[\mathcal{E}^{\star}_{\texttt{cc-not-cont}} \right] \leq 1/10$.
\end{claim}

\begin{proof}
By  Observation~\ref{obs:depth}, every  meta-node $x \in \mathcal{Z}^{\star}_{\texttt{cc-not-cont}}$ is at depth $\ell$. As there are $\binom{\Delta+1}{2} \leq \Delta^2$ possible types, the  meta-nodes at depth $\ell$    have at most $(\Delta^2)^{\ell+1} = \Delta^{2(\ell+1)}$ possible transcripts. For each  transcript, there are at most $n$ meta-nodes in $\mathcal{Z}^{\star}_{\texttt{cc-not-cont}}$ (see Lemma~\ref{lem:expansion:main}). Thus, we get:
$$|\mathcal{Z}^{\star}_{\texttt{cc-not-cont}}| \leq \Delta^{2(\ell+1)} n.$$

Next, consider any meta-node $x \in \mathcal{Z}^{\star}_{\texttt{cc-not-cont}}$, and let $(x_0, x_1, \ldots, x_{\ell})$ be the unique path from the root $r$ to $x$ in $\mathcal{T}$ (and in $\mathcal{T}^{\star}_{\ell}$), with $r = x_0$ and $x = x_{\ell}$. By our construction of the meta-tree $\mathcal{T}$, every internal meta-node on this path has exactly $L$ children in $\mathcal{T}$. Thus, the random walk taken by our algorithm traverses this path (and ends up at $x$) with probability  $1/L^{\ell}$. Summing these probabilities over all $x \in \mathcal{Z}^{\star}_{\texttt{cc-not-cont}}$, we get: $$\Pr\left[\mathcal{E}^{\star}_{\texttt{cc-not-cont}} \right] \leq |\mathcal{Z}^{\star}_{\texttt{cc-not-cont}}| \cdot (1/L^{\ell}) \leq \Delta^{2(\ell+1)} n/L^{\ell}  <  1/10,$$ where the last inequality follows from~(\ref{eq:parameters:main}).
\end{proof}

\begin{claim}
\label{cl:dirty}
We have $\Pr\left[\mathcal{E}^{\star}_{\texttt{d}} \right] \leq 1/5$.
\end{claim}

\begin{proof}
Consider any  $i \in [0, \ell-1]$. Let $\mathcal{E}_i$ denote the event that after $i$ recursive calls of Algorithm~\ref{alg:extend:main}, the concerned random walk is at a meta-node (say) $x_i$ that belongs to the core of $\mathcal{T}^{\star}$, and as a corollary, the event $\mathcal{E}^{\star}_{\texttt{d}}$  has not yet taken place. By Observation~\ref{obs:connect}, we have:
\begin{equation}
\label{eq:inference:0}
\Pr[\mathcal{E}_0] = 1.
\end{equation}
We will next prove the following inequality.
\begin{equation}
\label{eq:infer:main}
\Pr\left[ \, \overline{\mathcal{E}^{\star}_d} \cup \mathcal{E}_{i+1} \, \big| \, \mathcal{E}_i \, \right] \geq 1-\frac{1}{10\ell} \text{ for all } i \in [0, \ell-1].
\end{equation}
Towards this end, fix any $i \in [0, \ell-1]$, and condition on the event $\mathcal{E}_i$. Since  $x_i$ is part of the core of $\mathcal{T}^{\star}$, it is congenitally clean and not contaminated (see Observation~\ref{obs:internal}). Thus, the meta-node $x_i$ has $L$ children in $\mathcal{T}$ and at most $L/(10\ell)$ of these children are dirty. Accordingly, with probability at least $1-1/(10 \ell)$, in the very next step the random walk moves on to a non-dirty child (say) $y$ of $x$. Such a non-dirty child $y$ is either: (1) terminal, or (2) congenitally clean and contaminated, or (3) congenitally clean and not contaminated. In the former two cases, we are guaranteed that the event $\mathcal{E}^{\star}_{\texttt{d}}$ cannot occur, because the event $\mathcal{E}^{\star}_{\texttt{t}} \cup \mathcal{E}^{\star}_{\texttt{cc-cont}}$ has already taken place and this is mutually exclusive with the event $\mathcal{E}^{\star}_{\texttt{d}}$ (see Corollary~\ref{cor:partition}). We next consider the third case. Here, if $i < \ell-1$, then meta-node $y$ is still within the core of $\mathcal{T}^{\star}$, and hence the event $\mathcal{E}_{i+1}$ has occurred. Otherwise, if $i = \ell-1$, then the event $\mathcal{E}^{\star}_{\texttt{cc-not-cont}}$ has occurred which again is mutually exclusive with the event $\mathcal{E}^{\star}_d$ (see Corollary~\ref{cor:partition}). This concludes the proof of inequality~(\ref{eq:infer:main}).

From~(\ref{eq:inference:0}) and~(\ref{eq:infer:main}), we infer that
  $\Pr\left[\, \overline{\mathcal{E}^{\star}_d} \, 
\right] \geq (1-1/(10\ell))^{\ell} \geq 4/5$. Thus, we get $\Pr\left[ \, \mathcal{E}^{\star}_d \, \right] = 1 - \Pr\left[\, \overline{\mathcal{E}^{\star}_d}\, \right] \leq 1/5$.
\end{proof}

\begin{corollary}
\label{cor:prob}
We have $\Pr\left[\mathcal{E}^{\star}_{\texttt{t}}\right]
+
\Pr\left[\mathcal{E}^{\star}_{\texttt{cc-cont}}\right] \geq 7/10$.
\end{corollary}

\begin{proof}
Follows from Corollary~\ref{cor:partition}, Claim~\ref{cl:clean} and Claim~\ref{cl:dirty}.
\end{proof}

Let $\mathcal{E}^{\star}$ denote the event that the random walk taken by our algorithm  ends at a terminal meta-node at a depth $\leq \ell+1$. It is trivial to note that  $\Pr\left[\mathcal{E}^{\star} | \mathcal{E}^{\star}_{\texttt{t}}\right] = 1$. Henceforth, we condition on the event $\mathcal{E}^{\star}_{\texttt{cc-cont}}$. This means that the concerned random walk has reached some contaminated and congenitally clean meta-node $x$ (say) at depth $\leq \ell$. The meta-node $x$ has $L$ children in $\mathcal{T}$, and by Lemma~\ref{lem:cont is good:main} at least $1/(40\ell)$-fraction of its children are terminal. Thus, with probability at least $1/(40\ell)$, in the very next step the random walk moves on to a terminal child of $x$, at depth $\leq \ell+1$.  To summarize, we  deduce that $\Pr\left[\mathcal{E}^{\star} | \mathcal{E}^{\star}_{\texttt{t}} \cup \mathcal{E}^{\star}_{\texttt{cc-cont}} \right] \geq 1/(40 \ell)$. Lemma~\ref{lm:extension:main} now follows from~(\ref{eq:parameters:main}) and Corollary~\ref{cor:prob}.

\subsection{Getting Rid of Assumption~\ref{assume:main}: The Major Technical Hurdles}
\label{sec:algo:hurdle:main}

If we remove Assumption~\ref{assume:main}, our algorithm might now fail if, given an uncolored edge $e=(u,v)$, it finds a sufficiently short alternating path $P$ starting at $u$ and ending at $v$. In this case, our algorithm runs $\textsf{Apply}(\chi, e, P)$, but since $P$ has both $u$ \emph{and} $v$ as endpoints, this does not produce a proper coloring. In order to deal with this case, we need to use \emph{Vizing fans} in order to construct this alternating path $P$. In addition to making the algorithm more technical, this leads to the following two major hurdles that we need to overcome.

\medskip
\noindent \textbf{Hurdle 1:} Let $x$ and $y$ be two distinct congenitally clean meta-nodes with the same transcript $(\tau_0,\dots,\tau_i)$. Previously, by \Cref{cl:expansion:main}, we had that the paths $P_{\leq L}^{(x)}$ and $P_{\leq L}^{(y)}$ were vertex-disjoint. However, \emph{this is no longer necessarily the case}. Every time we construct a fan around some vertex $u$, we make changes to the colors of the edges around $u$, even if they have colors that are not contained in any of the types of the alternating paths that have been used so far. Thus, the alternating path $P^{(x)}$ of some congenitally clean meta-node $x$ might not be a maximal alternating path in the original coloring $\chi^{(r)}$. We say that such a meta-node $x$ is \emph{damaged}. Since each fan only changes the colors of at most $\Delta$ edges, each of which is contained in at most $O(\Delta)$ alternating paths, and our algorithm only runs for $\tilde O(1)$ steps, we can argue that each meta-node has at most $\tilde O(\Delta^2)$ many damaged children. By taking $L$ to be sufficiently large, we can ensure that (with probability $\Omega(1)$) we do not encounter any damaged meta-nodes in a random walk (see \Cref{lem:low damage}).

\medskip
\noindent \textbf{Hurdle 2:} Let $e=(u,v)$ be an uncolored edge and $c_u \in \miss_\chi(u)$, $c_v \in \miss_\chi(v)$ be blocking colors. Previously, our algorithm always found an alternating path $P$ starting at $u$ that did not use either $c_u$ or $c_v$. The ability to find an alternating path that can avoid such blocking colors is crucial for the proof of \Cref{lem:cont is good:main}. However, if we use Vizing fans to find alternating paths, then we can no longer guarantee that we can avoid using these blocking colors. In the case that we cannot avoid these blocking colors, we use a modified Vizing fan construction to find a $\{c_u,c_v\}$-alternating path that does not start at either $u$ or $v$ (see \Cref{lm:overlap}). Similarly to \Cref{lem:cont is good:main}, we consider a different case where an $\Omega(1/\ell)$-fraction of the children of a meta-node have this property. In this situation, we can guarantee that $\Omega(L/\ell)$ of the alternating paths corresponding to these children must be vertex-disjoint, and thus have an average length of $\tilde O(\sqrt{\Delta n})$, leading to $\Omega(L/\ell)$ many of these children being terminal.

\medskip
\noindent In \Cref{sec:extension}, we give the complete proof of \Cref{th:extension:main} without Assumption~\ref{assume:main}.

%% file: prelim.tex
In this part of the paper, we provide the full proof of \Cref{thm:main}, which we restate at the start of \Cref{sec:our alg}. In \Cref{sec:basic notation}, we define the basic notations used throughout the rest of the paper. In \Cref{sec:parts}, we describe the main algorithmic components that we use to construct our algorithm. In \Cref{sec:our alg}, we describe and analyse our algorithm. Finally, in Sections~\ref{sec:extension} and \ref{sec:stars}, we prove \Cref{thm:extension} and \Cref{lem:key 1} respectively. 

We note that this part of the paper is fully self-contained and uses slightly different terminology and notation than the extended abstract.

\section{Basic Notations}\label{sec:basic notation}
Let $G = (V, E)$ be an undirected simple graph on $n$ vertices and $m$ edges with maximum vertex degree $\Delta$. For any partial $(\Delta+1)$-edge coloring $\chi$ of $G = (V, E)$, we say that a neighbor $v$ of vertex $u$ is a {\em colored neighbor} (respectively, {\em uncolored neighbor}) of $u$ if $\chi(u, v)\neq\bot$ (resp., $\chi(u, v)= \bot$). For any vertex $u\in V$, let $\miss_\chi(u)\subseteq [\Delta+1]$ be the set of colors which are not used around $u$ by $\chi$. Let $N_\chi(u)$ denote the set of all colored neighbors of $u$ in $G$. We use $\deg_G^\chi(u)$ to denote the number of \emph{uncolored} edges incident on $u$ w.r.t.~$\chi$. For any pair of different colors $\{x, y\}$, a simple path $P$ between two endpoints $s,t$ is an $\{x, y\}$-alternating path under $\chi$, if all edges on $P$ are colored either $x$ or $y$; plus, $P$ is maximal if $\{x, y\}\cap \miss_\chi(s)\neq \emptyset$ and $\{x, y\}\cap \miss_\chi(t)\neq \emptyset$.

%% file: tools.tex
\section{The Algorithmic Components}\label{sec:parts}

In this section, we describe the subroutines that we use to construct our algorithm.

\subsection{Coloring Stars}

We prove the following lemma following the basic ideas from the recent work of \cite{BhattacharyaCCSZ24}. The proof of this lemma hinges on the same key technical observations but with a simpler algorithm and better runtime bound, whose full proof is deferred to \Cref{sec:stars}.

\begin{lemma}\label{lem:key 1}
There is an algorithm $\textnormal{\texttt{ColorLightStars}}$ that, given a graph $G$, a partial $(\Delta + 1)$-edge coloring $\chi$ of $G$ and subset $U^\star \subseteq V$ such that:
\begin{itemize}
    \item $|\miss_\chi(u)| \geq d$ for all $u \in U^\star$ for some positive integer $d \in \mathbb N$,
    \item there are $\lambda$ uncolored edges incident on $U^\star$,
\end{itemize}
extends the coloring $\chi$ to $\Omega(\lambda)$ uncolored edges incident on $U^\star$ in $\tilde O(\lambda \Delta + \Delta m / d)$ expected time.
\end{lemma}
\noindent {\bf Remark.} As mentioned in the technical overview of \Cref{sec:overview}, previously in \cite{BhattacharyaCCSZ24}, the runtime bound for the same task is $\tilde{O}\brac{\lambda\Delta + \min_{\tau\geq 1}\{\frac{\Delta m\tau}{d} + \frac{\lambda n}{\tau}\}}$, which is always worse than the time bound of \Cref{lem:key 1}.
To clarify, note that in the technical overview of \cite{BhattacharyaCCSZ24}, the authors claimed the same runtime bound $\tilde O(\lambda \Delta + \Delta m / d)$ in their Lemma 2.1 as we claim in the current paper (see \Cref{lem:key 2}, which we derive from \Cref{lem:key 1}), but their Lemma 2.1 was only given as part of an informal technical overview, provided  merely for the sake of conveying the high-level ideas. 
What \cite{BhattacharyaCCSZ24} actually achieved implicitly (in the formal part of their paper) is the weaker runtime bound of $\tilde{O}\brac{\lambda\Delta + \min_{\tau\geq 1}\{\frac{\Delta m\tau}{d} + \frac{\lambda n}{\tau}\}}$, which was sufficient for their $\tilde{O}(mn^{1/3})$ time algorithm, but it is insufficient for our goal of achieving a runtime of $\tilde{O}(mn^{1/4})$.\\

\noindent In our final algorithm, we use \texttt{ColorLightStars} to extend the coloring $\chi$ to vertices with sufficiently low degree. To deal with high degree vertices, we also need a variant of the algorithm that we call \texttt{ColorHeavyStars}, which follows as a corollary from \Cref{lem:key 1}.

\begin{corollary}\label{lem:key 2}
There is an algorithm $\textnormal{\texttt{ColorHeavyStars}}$ that, given a graph $G$, a partial $\Delta + 1$ edge coloring $\chi$ of $G$ and subset $U^\star\subseteq V$ such that there are $\lambda$ uncolored edges incident on $U^\star$,
extends the coloring $\chi$ to $\Omega(\lambda)$ uncolored edges incident on $U^\star$ in $\tilde O(\lambda \Delta + \Delta m |U^\star| / \lambda)$ expected time.
\end{corollary}

\begin{proof}
    For each integer $p$, let $U^\star_p := \{ u \in V \mid \deg_G^\chi(u) \in [2^p, 2^{p+1})\}$.\footnote{Recall that $\deg_G^\chi(u)$ denotes the number of uncolored edges incident on $u$ w.r.t.~$\chi$.} Let $p \in \mathbb N$ be the value that maximizes $\sum_{u \in U^\star_p} \deg_G^\chi(u)^2$. Let $\lambda'$ denote the number of uncolored edges incident on $U^\star_p$ and $d = 2^p$. Then it follows that
    \begin{eqnarray} \nonumber
    \lambda' \cdot 2d &=& \sum_{u \in U^\star_p} \deg_G^\chi(u) \cdot 2d \geq \sum_{u \in U^\star_p} \deg_G^\chi(u)^2 \geq \frac{1}{\log n} \cdot \sum_{u \in U^\star} \deg_G^\chi(u)^2 
     \\ &\geq& \frac{1}{\log n} \cdot \left( \frac{1}{\sqrt{|U^\star|}} \cdot \sum_{u \in U^\star} \deg_G^\chi(u) \right)^2 = \frac{1}{\log n} \cdot \frac{\lambda^2}{|U^\star|}, 
    \end{eqnarray}
    where the last inequality follows from the Cauchy-Schwarz inequality. 
    It follows that $\lambda^2 / (\lambda' |U^\star|) \leq d \cdot 2\log n$. The set $U^\star_p$ has the properties that
    \begin{enumerate}
    \item $|\miss_\chi(u)| \geq d$ for all $u \in U^\star_p$,
    \item there are $\lambda'$ uncolored edges incident on $U^\star_p$.
\end{enumerate}
Thus, applying the algorithm $\textnormal{\texttt{ColorLightStars}}$ (\Cref{lem:key 1}) with the set $U^\star_p$, we can color $\Omega(\lambda')$ uncolored edges incident on $U^\star_p$ in expected time 
\begin{equation}\label{eq:k1 to k2}
    \tilde O \! \left(\lambda' \Delta + \frac{\Delta m}{d} \right) 
\leq \tilde O \! \left(\lambda' \Delta + \frac{\Delta m |U^\star|}{\lambda} \cdot \frac{\lambda'}{\lambda} \right) \leq \tilde O \! \left(\lambda \Delta 
+ \frac{\Delta m |U^\star|}{\lambda}\right) \cdot \frac{\lambda'}{\lambda}.
\end{equation}
We can repeat this process until the total number of edges to which we have extended the coloring exceeds $\lambda/2$. Let $\lambda_i$ denote the number of uncolored edges at the start of the $i$ iteration and let $\lambda_i'$ denote the number of edges that we extend the coloring to in the $i^{th}$ iteration. Suppose that we perform $t$ iterations in total. Then, for all $i \in [t]$, we have that $\lambda/2 \leq \lambda_i \leq \lambda$. It follows from \Cref{eq:k1 to k2} that the expected total running time of our algorithm is at most
$$ \sum_{i=1}^t  \tilde O \! \left(\lambda_i \Delta 
+ \frac{\Delta m |U^\star|}{\lambda_i}\right) \cdot \frac{\lambda'_i}{\lambda_i} \leq  \tilde O \! \left(\lambda \Delta 
+ \frac{\Delta m |U^\star|}{\lambda}\right) \cdot \left( \frac{1}{\lambda} \cdot \sum_{i=1}^t \lambda'_i \right) \leq \tilde O \! \left(\lambda \Delta 
+ \frac{\Delta m |U^\star|}{\lambda}\right).$$
Thus, the expected total time spent handling these calls to $\textnormal{\texttt{ColorLightStars}}$ is upper bounded by $\tilde O(\lambda \Delta + \Delta m |U^\star|/\lambda)$.
\end{proof}

\subsection{Near-Vizing Coloring}

The following theorem is the main technical result in the recent work of \cite{Assadi24}. While this is not strictly necessary to prove \Cref{thm:assadi}, it significantly simplifies the technical details of our algorithm by allowing us to avoid using Euler partitioning and recursion.

\begin{theorem}\label{thm:assadi}
There is an algorithm $\textnormal{\texttt{NearVizingColoring}}$ that, given a graph $G$, computes a $\Delta + 300 \log n$ edge coloring $\chi$ of $G$ in $\tilde O(m)$ time with high probability.
\end{theorem}

\subsection{Color Extension Theorem}

The following theorem is our main technical contribution. This theorem is a generalization of the results of \cite{duan2019dynamic} to the case where we only have access to $\Delta + 1$ colors instead of $\Delta + \tilde \Omega(\sqrt{ \Delta})$. We defer the proof of this theorem to \Cref{sec:extension}.

\begin{theorem}\label{thm:extension}
There is an algorithm $\textnormal{\texttt{FastVizingExtension}}$ that, given a graph $G$, a partial $\Delta + 1$ edge coloring $\chi$ of $G$ and an uncolored edge $e$, extends the coloring $\chi$ to the edge $e$ in $\tilde O(\min\{\Delta^2 + \sqrt{\Delta n}, n\})$ expected time.
\end{theorem}

%% file: alg.tex
\section{Our Algorithm}\label{sec:our alg}
\label{sec:full:algo}

In this section, we prove \Cref{thm:main}, which we restate below.

\begin{theorem}[\Cref{thm:main} Restated]
Given a simple, undirected $m$-edge $n$-node graph $G = (V, E)$ with maximum degree $\Delta$, we can compute a $(\Delta + 1)$-edge coloring of $G$ in $\tilde O(mn^{1/4})$ time with high probability.
\end{theorem}

Let $G = (V,E)$ be a graph with maximum degree $\Delta$. We assume that $\Delta \geq \kappa n^{1/4} \log n$, where $\kappa := 10^4$ is a constant. Otherwise, we can use the $\tilde O(m\Delta)$ time algorithm of \cite{sinnamon2019fast} to $\Delta + 1$ edge color $G$ in $\tilde O(mn^{1/4})$ time. Let $V_{\textsf{lo}} := \{u \in V \mid \deg_G(u) \leq \Delta/2 \}$ be the {\em low degree} nodes, $V_{\textsf{hi}} := V \setminus V_{\textsf{lo}}$ be the {\em high degree} nodes, $n_{\textsf{lo}} := |V_{\textsf{lo}}|$ and $n_{\textsf{hi}} := |V_{\textsf{hi}}|$. Given the graph $G$, we now describe how our algorithm computes a $(\Delta + 1)$-edge coloring $\chi$ of the graph $G$ in 2 phases.

\subsubsection*{Phase 1: Extracting Stars}

Our algorithm begins by finding a subset $U \subseteq V$ that satisfies the following:
\begin{equation}\label{eq:good set}
    \text{(I) } \Delta(G[V \setminus U]) \leq \Delta - 300 \log n, \;  \text{(II) } |U_{\textnormal{\textsf{hi}}}| \leq 
    \frac{|V_{\textnormal{\textsf{hi}}}|}{\Delta} \cdot \frac{3\kappa\log n}{2}, \;  \text{(III) } \lambda_{\textsf{init}} \leq \ \frac{m}{\Delta} \cdot 10\kappa \log n,
\end{equation}
where $U_{\textnormal{\textsf{hi}}} := U \cap V_{\textnormal{\textsf{hi}}}$ and $\lambda_{\textsf{init}} := \sum_{u \in U} \deg_G(u)$.
We use the following algorithm in order to obtain such a set $U$. 

\medskip
\begin{algorithm}[H]
    \SetAlgoLined
    \DontPrintSemicolon
    \For{$10\log n$ \textnormal{\textbf{iterations}}}{
        Sample $U \subseteq V$ by placing each $u \in V$ into $U$ independently with probability $\kappa \log n/ \Delta$\;
        \If{$|V_{\textnormal{\textsf{hi}}}| < \Delta / 4$}{
            $U \leftarrow U \cap V_{\textnormal{\textsf{lo}}}$\;
        }
        \If{$U$ satisfies \Cref{eq:good set}}{
            \Return{$U$}
        }
    }
    \caption{\texttt{ExtractStars}$(G)$}
    \label{alg:extractstars}
\end{algorithm}
\medskip
\noindent
The following lemma summarises the behaviour of \Cref{alg:extractstars}.
\begin{lemma}\label{lem:extract stars}
    \Cref{alg:extractstars} runs in $\tilde O(m)$ time and returns a set $U$ satisfying \Cref{eq:good set} w.h.p.
\end{lemma}

\begin{proof}
    Each iteration of the algorithm can be implemented to run in $\tilde O(m)$ time in the obvious way.
    We now argue that each iteration of the algorithm finds a set $U$ which satisfies \Cref{eq:good set} with probability at least $1/2$ by lower bounding the probability that the set $U$ satisfies each of the conditions.

    \medskip
    \noindent \textbf{Condition (I):} We first consider the case that $|V_{\textnormal{\textsf{hi}}}| \geq \Delta / 4$. Given some $u \in V$, we want to show that $\deg_{V \setminus U}(u) \leq \Delta - 300\log n$ w.h.p. If $\deg_V(u) \leq \Delta/2$, then this holds trivially since $\Delta \geq \kappa \log n$. Otherwise, if $\deg_V(u) \geq \Delta/2$, we have that $\mathbb E [\deg_U(u)] = \deg_V(u) \cdot (\kappa / \Delta) \log n \geq (\kappa / 2) \log n$ and it follows from Chernoff bounds that 
    $$\Pr \left[\deg_U(u) < \frac{\kappa}{4} \cdot \log n \right] \leq \exp \left( - \frac{\kappa}{16} \cdot \log n \right) \leq n^{-\kappa / 16}.$$
    Thus, $\deg_U(u) \geq 300 \log n$ with probability at least $1 - n^{-\kappa/16}$. Taking a union bound, we get that Condition (I) holds w.h.p.
    Now, consider the case that $|V_{\textnormal{\textsf{hi}}}| < \Delta / 4$. Given some $u \in V$, we want to show that $\deg_{V \setminus U_{\textsf{lo}}}(u) \leq \Delta - 300\log n$ w.h.p., where $U_{\textsf{lo}}(u) = U \cap 
    V_{\textsf{lo}}(u)$.
    If $\deg_V(u) \leq \Delta/2$, then this again holds trivially. Otherwise, $u$ has at least $\Delta / 4$ neighbors in $V_{\textsf{lo}}$. It follows that $\mathbb E [\deg_{U_{\textsf{lo}}}(u)] \geq (\kappa / 4) \log n$ and, applying Chernoff bounds, we get that 
    $$\Pr \left[\deg_{U_{\textsf{lo}}}(u) < \frac{\kappa}{8} \cdot \log n \right] \leq \exp \left( - \frac{\kappa}{32} \cdot \log n \right) \leq n^{-\kappa / 32}.$$ 
    Taking a union bound, we again get that Condition (I) holds w.h.p.
    
    \medskip
    \noindent \textbf{Condition (II):} If $|V_{\textnormal{\textsf{hi}}}| < \Delta / 4$, then $U_{\textsf{hi}} = \emptyset$, and Condition (II) holds trivially. If $|V_{\textnormal{\textsf{hi}}}| \geq \Delta / 4$, we have that $\mathbb E[|U_{\textsf{hi}}|] = (|V_{\textsf{hi}}|/\Delta) \cdot \kappa \log n \geq (\kappa/4) \log n$. Applying Chernoff bounds, we get that 
    $$ \Pr \left[|U_{\textsf{hi}}| \geq \frac{3}{2} \cdot \frac{|V_{\textsf{hi}}|}{\Delta} \cdot \kappa \log n \right] \leq \exp \left( - \frac{1}{10} \cdot \frac{|V_{\textsf{hi}}|}{\Delta} \cdot \kappa \log n \right) \leq \exp \left( - \frac{\kappa}{40} \cdot \log n \right) \leq n^{-\kappa / 40}. $$

    \medskip
    \noindent \textbf{Condition (III):} We first note that
    $$\mathbb E \left[ \sum_{u \in U} \deg_G(u) \right] = \sum_{u \in V} \deg_G(u) \cdot \Pr [u \in U] = \frac{2m\kappa \log n}{\Delta}. $$
    By Markov's inequality, it follows that
    $$\Pr \left[ \sum_{u \in U} \deg_G(u) \geq \frac{10m\kappa \log n}{\Delta} \right] \leq \frac{1}{5}.$$

    \noindent
    By taking a union bound, we have that $U$ satisfies \Cref{eq:good set} with probability at least $1/2$.
\end{proof}

\subsubsection*{Phase 2: Coloring the Graph}

After finding a set $U$ that satisfies \Cref{eq:good set}, our algorithm proceeds to edge coloring the graph $G$. Using the fact that $U$ satisfies Condition~(I), we can apply \texttt{NearVizingColoring} (\Cref{thm:assadi}) in order to produce a $(\Delta + 1)$-edge coloring $\chi$ of $G[V \setminus U]$. We then use \texttt{ColorLightStars} (\Cref{lem:key 1}) to extend $\chi$ to all of the edges incident on nodes in $U_{\textsf{lo}} := \{u \in U \mid \deg_G(u) \leq \Delta/2 \}$. Finally, we use \texttt{ColorHeavyStars} (\Cref{lem:key 2}) to extend $\chi$ to all of the edges incident on nodes in $U_{\textsf{hi}} := U \setminus U_{\textsf{lo}}$.
\Cref{alg:final} gives the full description of the second phase of our algorithm using the subroutines from \Cref{sec:parts}. We write $n_{\textsf{hi}} = |V_{\textsf{hi}}|$.

\medskip
\begin{algorithm}[H]
    \SetAlgoLined
    \DontPrintSemicolon
    \SetKwRepeat{Do}{do}{while}
    \SetKwBlock{Loop}{repeat}{EndLoop}
    \texttt{// Step 1:~Efficiently color the non-star edges using slack}\;
    $\chi \leftarrow \texttt{NearVizingColoring}(G[V \setminus U])$\;
    \medskip
    \texttt{// Step 2:~Extend $\chi$ to the light stars}\;
    $\lambda_{{\textsf{lo}}} := \sum_{u \in U_{\textsf{lo}}} \deg_G^\chi(u)$\;
    \While{$\lambda_{{\textnormal{\textsf{lo}}}} \geq 1$}{
        $\chi \leftarrow \texttt{ColorLightStars}(G, \chi, U_{\textsf{lo}})$\;
    }
    \medskip
    \texttt{// Step 3:~Extend $\chi$ to the heavy stars}\;
    $\lambda_{{\textsf{hi}}} := \sum_{u \in U_{\textsf{hi}}} \deg_G^\chi(u)$\;
    $\tau := \sqrt{m n_{\textsf{hi}} / \min \{\Delta^2 + \sqrt{\Delta n}, n\}}$\;
    \While{$\lambda_{{\textnormal{\textsf{hi}}}} > \tau$}{
        $\chi \leftarrow \texttt{ColorHeavyStars}(G, \chi, U_{\textsf{hi}})$\;
    }
    \While{$\lambda_{{\textnormal{\textsf{hi}}}} \geq 1$}{
        Let $e$ be an uncolored edge\;
        $\chi \leftarrow \texttt{FastVizingExtension}(G, \chi, e)$\;
    }
    \Return{$\chi$}
    \caption{\texttt{FastColoring}$(G, U)$}
    \label{alg:final}
\end{algorithm}
\medskip

\noindent The following theorem summarizes the properties of \Cref{alg:final}.

\begin{theorem}\label{thm:alg}
    \Cref{alg:final} produces a $(\Delta + 1)$-edge coloring of the graph $G$ in expected time $\tilde O(mn^{1/4})$.
\end{theorem} 
\Cref{thm:main} follows directly from \Cref{thm:alg} and \Cref{lem:extract stars} using standard probabilistic amplification. 

\subsection{Analysis}

It follows from the properties of the set $U$ and the correctness of the subroutines used by our algorithm that \Cref{alg:final} always returns a $(\Delta + 1)$-edge coloring of $G$.
It remains to bound the expected running time of our algorithm. It follows directly from \Cref{thm:assadi} that Step~1 of our algorithm takes $\tilde O(m)$ expected time. We now show that Steps~2 and 3 of our algorithm take $\tilde O(m)$ and $\tilde O(mn^{1/4})$ expected time, respectively.

\subsubsection*{The Running Time of Step 2}

We first note that, for all $u \in U_{\textsf{lo}}$, we have  $|\miss_G^\chi(u)| \geq \Delta/2$. Applying \Cref{lem:key 1}, it follows that Step~2 of our algorithm runs for $\tilde O(1)$ iterations, where each iteration takes expected time $\tilde O(\lambda_{\textsf{lo}} \Delta + m)$. By the properties of the vertex set $U$ (refer to \Cref{eq:good set}) obtained by
\Cref{alg:extractstars}, we have $\lambda_{\textsf{lo}} \leq \lambda_{\textsf{init}} \leq \tilde O(m/\Delta)$. Thus, Step~2 takes $\tilde O(m)$ expected time in total.

\subsubsection*{The Running Time of Step 3}

Applying \Cref{lem:key 2} and \Cref{thm:extension}, it follows that Step~3 of our algorithm takes expected time
$$ \tilde O \! \left(\lambda_{\textsf{hi}}\Delta + \frac{\Delta m|U_{\textsf{hi}}|}{\tau} \right) + \tilde O \!\left(\tau \cdot \min \{\Delta^2 + \sqrt{\Delta n}, n\}\right). $$
By the properties of the vertex set $U$ (refer to \Cref{eq:good set}),
we have $\lambda_{\textsf{hi}} \leq \lambda_{\textsf{init}} \leq \tilde O(m/\Delta)$ and $|U_{\textsf{hi}}| \leq \tilde O(n_{\textsf{hi}}/\Delta)$, hence the expected running time of Step 3 of our algorithm is bounded by
\begin{eqnarray*} \tilde O \! \left(m + \frac{mn_{\textsf{hi}}}{\tau} \right) + \tilde O \!\left(\tau \cdot \min \{\Delta^2 + \sqrt{\Delta n}, n\}\right)  &=& \tilde O(m) + \tilde O \! \left(\sqrt{mn_{\textsf{hi}} \cdot \min\{\Delta^2 + \sqrt{\Delta n}, n\}}\right) 
\\ &=& \tilde O(m) + \tilde O \! \left(\sqrt{mn_{\textsf{hi}}} \cdot \min\{\Delta + (\Delta n)^{1/4}, \sqrt n\}\right). 
\end{eqnarray*}
Let $\Gamma := \sqrt{mn_{\textsf{hi}}} \cdot \min\{\Delta + (\Delta n)^{1/4}, \sqrt n\}$. We now show that $\Gamma \leq 4 mn^{1/4}$, which implies that Step~3 has an expected running time of $\tilde O(mn^{1/4})$.
We consider the following 3 cases.

\medskip
\noindent \emph{Case 1: $\Delta \geq \sqrt n$.} Then $\min\{\Delta + (\Delta n)^{1/4}, \sqrt n\} \leq \sqrt n$, thus 
$$\Gamma \leq \sqrt{mn_{\textsf{hi}}} \cdot \sqrt n = m \sqrt{n_{\textsf{hi}}n/m} < m \sqrt{4n/\Delta} = 2 m \sqrt{n/\Delta} \leq 2 mn^{1/4},$$ 
where the penultimate and last inequalities hold as
$\Delta n_{\textsf{hi}} \leq 4m$ and $\Delta \geq \sqrt{n}$, respectively.

\medskip
\noindent \emph{Case 2: $n^{1/3} \leq \Delta < \sqrt n$.} Then $\min\{\Delta + (\Delta n)^{1/4}, \sqrt n\} \leq 2\Delta$, thus  
$$\Gamma \leq 2\sqrt{mn_{\textsf{hi}}} \cdot \Delta = 2\sqrt{m\Delta n_{\textsf{hi}}} \cdot \sqrt{\Delta} \le 4m \sqrt \Delta \leq 4mn^{1/4},$$ 
where the penultimate and last inequalities hold as $\Delta n_{\textsf{hi}}\leq 4m$ and $\Delta \leq \sqrt n$.

\medskip
\noindent \emph{Case 3: $\Delta < n^{1/3}$.} Then $\min\{\Delta + (\Delta n)^{1/4}, \sqrt n\} \leq 2(\Delta n)^{1/4}$, thus 
$$\Gamma \leq 2\sqrt{mn_{\textsf{hi}}} \cdot (\Delta n)^{1/4} = 2m \sqrt{n_{\textsf{hi}}/m} \cdot (\Delta n)^{1/4} \leq 2m \sqrt{4/\Delta} \cdot (\Delta n)^{1/4} = 4m (n/\Delta)^{1/4} \leq 4mn^{1/4},$$
where the penultimate inequality holds as $\Delta n_{\textsf{hi}} \leq 4m$.

%% file: extension.tex
\section{Color Extension to Edges}\label{sec:extension}

In this section, we prove the following theorem.

\begin{theorem}[\Cref{thm:extension} Restated]\label{thm:extension expected}
    Given a undirected simple graph $G = (V, E)$ on $n$ vertices and maximum degree $\Delta$, as well as and a partial $(\Delta + 1)$-edge coloring $\chi$ of $G$ with an uncolored edge $e$, we can extend $\chi$ to the edge $e$ in $\tilde O(\Delta^2 + \sqrt{\Delta n})$ expected time.
\end{theorem}

This section is self-contained (excluding the basic notations from \Cref{sec:basic notation}) and uses slightly different notation than the sketch of this proof in \Cref{sec:th:extension:main}.

\subsection{Preliminaries}

We now describe the main components used in our multi-step Vizing chain construction.

\subsubsection*{The Algorithm \textnormal{\textsf{VizingFan}}}

The first component of our algorithm is a standard Vizing fan construction \cite{sinnamon2019fast, duan2019dynamic} but with one simple modification. The algorithm $\VizingF$ takes as input some uncolored edge $(u,v)$ and colors $c_u \in \miss_\chi(u)$ and $c_v \in \miss_\chi(v)$. It then proceeds to either extend $\chi$ to the uncolored edge $(u,v)$ by shifting colors around $u$, or constructs a Vizing fan around the vertex $u$ while ensuring that (1) none of the vertices $v_i$ participating in the fan have $c_u \in \miss_\chi(v_i)$,
and (2) that the color of the first vertex in the fan is not $c_v$. 
\Cref{alg:vizing fan} gives a formal description of this procedure.

\medskip

\begin{algorithm}[H]
    \SetAlgoLined
    \DontPrintSemicolon
    \SetKwRepeat{Do}{do}{while}
    \SetKwBlock{Loop}{repeat}{EndLoop}
    $i \leftarrow 0$ and $v_0 \leftarrow v$\;
    Let $c_0 \in \miss_\chi(v_0) \setminus \{c_v\}$\;
    \While{$c_i \notin \{c_0 ,\dots, c_{i-1}\}$ \textnormal{and} $c_i \notin \miss_\chi(u)$}{
        Let $(u, v_{i+1})$ be the edge with color $\chi(u, v_{i+1}) = c_i$\;
        Let $c_{i+1} \in \miss_\chi(v_{i+1})$\;
        \If{$c_u \in \miss_\chi(v_{i+1})$}{
            $c_{i+1} \leftarrow c_u$\;
        }
        $i \leftarrow i + 1$\;
    }
    \If{$c_i \in \miss_\chi(u)$}{
        \For{$j = 0\dots i$}{
            $\chi(u, v_j) \leftarrow c_j$\;
        }
        \Return{$\chi$}
    }
    \Else{
        \Return{$(v_0,c_0),\dots,(v_i,c_i)$}
    }
    \caption{\textsf{VizingFan}$(\chi,u, v, c_u, c_v)$}
    \label{alg:vizing fan}
\end{algorithm}

\medskip
\noindent
We refer to the sequence $(v_0,c_0),\dots,(v_i,c_i)$ returned by \Cref{alg:vizing fan} as a Vizing fan. We say that a Vizing fan is $c$-primed if $c = c_i$.

\subsubsection*{The Algorithm \textnormal{\textsf{Vizing}}}

We now give an algorithm $\Vizing$ which builds a Vizing chain using $\VizingF$ as a subroutine and extends the coloring $\chi$ to the uncolored edge $(u,v)$. The algorithm also takes colors $c_u \in \miss_\chi(u)$ and $c_v \in \miss_\chi(v)$ as input.
We make some small modifications to the standard implementation of this algorithm for technical reasons that will become clear later on. \Cref{alg:vizing} gives a formal description of this procedure. 

\begin{algorithm}[H]
    \SetAlgoLined
    \DontPrintSemicolon
    \SetKwRepeat{Do}{do}{while}
    \SetKwBlock{Loop}{repeat}{EndLoop}
    \If{\textnormal{$\textsf{VizingFan}(\chi,u, v, c_u, c_v)$ extends $\chi$}}{
        $\chi \leftarrow \textsf{VizingFan}(\chi,u, v, c_u, c_v)$\;
        \Return{$\chi$}
    }
    $F = (v_0,c_0),\dots,(v_k,c_k) \leftarrow \textsf{VizingFan}(\chi,u, v, c_u, c_v)$\;
    \If{\textnormal{$F$ is \emph{not} $c_v$-primed}}{
        Let $c \in \miss_\chi(u) \setminus \{c_u\}$\;
        Let $P$ denote the maximal $\{c, c_k\}$-alternating path starting at $u$\;
    }
    \If{\textnormal{$F$ is $c_v$-primed}}{
        Let $c_i$ be the first appearance of $c_v$ in $c_0,\dots,c_k$\;
        Let $P$ denote the maximal $\{c_u, c_v\}$-alternating path starting at $v_i$\;
    }
    Extend $\chi$ to $(u,v)$ by flipping the path $P$ and shifting colors around $u$ (see \Cref{lem:full vizing proof})\;
    \Return{$\chi$}
    \caption{\textsf{Vizing}$(\chi,u, v, c_u, c_v)$}
    \label{alg:vizing}
\end{algorithm}

\medskip
\noindent Thus, running $\textsf{Vizing}(\chi,u, v, c_u, c_v)$ extends the coloring $\chi$ to the uncolored edge $(u,v)$ by shifting colors around the vertex $u$ and flipping the colors of the alternating path $P$. The analysis of the case where $F$ is not $c_v$-primed is the same as in the standard Vizing chain construction \cite{sinnamon2019fast}. The case where $F$ is $c_v$-primed follows from a similar argument. The following lemma summarises the behaviour of this algorithm.

\begin{lemma}\label{lem:full vizing proof}
    \Cref{alg:vizing} extends the coloring $\chi$ to the edge $(u,v)$ in time $\tilde O(\Delta + |P|)$.
\end{lemma}

\begin{proof}
    We begin by showing that the path $P$ is well-defined. In the case that $F$ is not $c_v$-primed, we know that $c \in \miss_\chi(u)$ and $c_k \notin \miss_\chi(u)$, so there is a $\{c,c_k\}$-alternating path starting at $u$. In the case that $F$ is $c_v$-primed, we know by the properties of our Vizing fan construction that $c_v \in \miss_\chi(v_i)$ and $c_u \notin \miss_\chi(v_i)$, so there is a $\{c_u,c_v\}$-alternating path starting at $u$. Thus, in both cases, the path $P$ is well-defined. We now describe how to extend the coloring $\chi$ to $(u,v)$ in both cases.
    
    If $F$ is not $c_v$-primed: Let $c_j$ be the first appearance of $c_k$. We consider the cases where the $P$ does or does not have $v_{j}$ as an endpoint. If $P$ does not end at $v_j$, then we can shift the colors of the first $j + 1$ edges in the fan by setting $\chi(u, v_0) \leftarrow c_0,\dots,\chi(u, v_{j}) \leftarrow c_{j}$ and flip the colors of the alternating path $P$.
    If $P$ does end at $v_j$, then we flip the colors of the alternating path $P$, shift the colors of the fan by setting $\chi(u, v_0) \leftarrow \chi(u, v_1),\dots,\chi(u, v_{k-1}) \leftarrow \chi(u, v_{k})$, and set $\chi(u, v_k) \leftarrow c_k$. Note that, while shifting the fan, we have $\chi(u, v_{i+1}) = c \neq c_{k}$.

    If $F$ is $c_v$-primed: We consider the cases where the $P$ does or does not have $u$ as an endpoint. If $P$ does not end at $u$, then the edge $(u, v_{i+1})$ (which has color $c_v$) is not contained in $P$. Thus, we can shift the colors of the first $i$ edges in the fan by setting $\chi(u, v_0) \leftarrow c_0,\dots,\chi(u, v_{i-1}) \leftarrow c_{i-1}$, flip the colors of the alternating path $P$, and set $\chi(u, v_i) \leftarrow c_u$.
    If $P$ does end at $u$, then the edge $(u, v_{i+1})$ is contained in $P$. Thus, we can flip the colors of the alternating path $P$, shift the colors of the fan by setting $\chi(u, v_0) \leftarrow \chi(u, v_1),\dots,\chi(u, v_{k-1}) \leftarrow \chi(u, v_{k})$, and set $\chi(u, v_k) \leftarrow c_v$.  Note that, while shifting the fan, we have $\chi(u, v_{i+1}) = c_u \neq c_i$.

    Using standard data structures, we can implement this algorithm to run in time $\tilde O(\Delta + |P|)$.
\end{proof}

The following lemma summarises the key properties of the path $P$ considered by \Cref{alg:vizing}.

\begin{lemma}
\label{lm:overlap}
    The path $P$ considered by the algorithm satisfies one of the following properties:
\begin{enumerate}
    \item $P$ is a maximal $\{c',c''\}$-alternating path in $\chi$ starting at $u$ where $\{c',c''\} \cap \{c_u, c_v\} = \emptyset$.\label{type 1 path}
    \item $P$ is a maximal $\{c_u, c_v\}$-alternating path in $\chi$ starting at neither $u$ nor $v$.\label{type 2 path}\footnote{Note that, in this case, the path $P$ starts at the vertex $v_i$ of the fan such that $c_i$ is the first appearance of $c_v$ in $c_0,\dots,c_k$. Since $c_0 \neq c_v$, we know that $v_i \neq v$.}
\end{enumerate}
\end{lemma}
These properties of the path $P$ will be crucial in our analysis later on. We refer to an alternating path satisfying Condition \ref{type 1 path} (resp.~Condition \ref{type 2 path}) above as \emph{non-overlapping} (resp.~\emph{overlapping}).

\subsubsection*{The Algorithm \textnormal{\textsf{TruncatedVizing}}}

We now give an algorithm \textsf{TruncatedVizing}, which takes the same input as \textsf{Vizing} along with an additional argument $t \in \mathbb N$.
It then proceeds to extend the coloring $\chi$ to the edge $e$ in the same way as \textsf{Vizing}, with one difference: If the maximal alternating path $P$ that is flipped by \textsf{Vizing} has length greater than $t$, the algorithm only flips the first $t-1$ edges of $P$ and leaves the $t^{th}$ edge in the path $P$ uncolored. It then returns the new coloring obtained after applying this procedure along with the edge left uncolored (as long as $P$ had length greater than $t$). \Cref{alg:trunc vizing} gives a formal description of this procedure.

\medskip
\begin{algorithm}[H]
    \SetAlgoLined
    \DontPrintSemicolon
    \SetKwRepeat{Do}{do}{while}
    \SetKwBlock{Loop}{repeat}{EndLoop}
    Let $P=e_1,\dots,e_t$ denote the alternating path considered by $\textsf{Vizing}(\chi,u, v, c_u, c_v)$\;
    \If{$|P| \leq t$}{
        $\chi \leftarrow \textsf{Vizing}(\chi,u, v, c_u, c_v)$\;
        \Return{$\chi$}\;
    }
    \If{\textnormal{$P$ is non-overlapping}}{
        Let $(u, v_i)$ be the first edge in $P$\;
        \For{$j = 0\dots i-1$}{
            $\chi(u, v_j) \leftarrow c_j$\;
        }
    }
    \If{\textnormal{$P$ is overlapping}}{
        Let $v_i$ be the first vertex in $P$\;
        \For{$j = 0\dots i-1$}{
            $\chi(u, v_j) \leftarrow c_j$\;
        }
        $\chi(u, v_i) \leftarrow c_u$\;
    }
    Flip the colors of the first $t-1$ edges in the $\{c',c''\}$-alternating path $P$\;
    $\chi(e_t) \leftarrow \bot$\;
    $(u',v') \leftarrow e_t$\;
    Let $c_u' \in \miss_\chi(u') \cap \{c', c''\}$ and $c_v' \in \miss_\chi(v') \cap \{c', c''\}$\;
    \Return{$\chi, u', v', c_u', c_v'$}
    \caption{\textsf{TruncatedVizing}$(\chi,u, v, c_u, c_v, t)$}
    \label{alg:trunc vizing}
\end{algorithm}

\subsection{Our Algorithm}

We define parameters $\ell := 10^2 \log n$ and $L := 10^3\ell^2(\Delta^{2} + \sqrt{\Delta n})$.
Our algorithm is given an uncolored edge $(u,v)$ and starts by attempting to construct a Vizing chain by calling $\Vizing$. If the Vizing chain has length $\Omega(L)$, then our algorithm instead calls $\TVizing$ and randomly truncates the Vizing chain after $O(L)$ steps. It then repeats this process, moving the uncolored edge around the graph by randomly truncating long Vizing chains, until it finds a short Vizing chain and successfully extends the coloring.
\Cref{lem:extend works} gives a formal description of this procedure.

\begin{algorithm}[H]
    \SetAlgoLined
    \DontPrintSemicolon
    \SetKwRepeat{Do}{do}{while}
    \SetKwBlock{Loop}{repeat}{EndLoop}
    Let $c_u \in \miss_\chi(u)$ and $c_v \in \miss_\chi(v)$\;
    \Loop{
    Let $P = u_1,\dots,u_{k+1}$ be the alternating path considered by $\textsf{Vizing}(\chi,u, v, c_u, c_v)$\;
    \If{$k \leq 2L+2$}{
            $\chi \leftarrow \textsf{Vizing}(\chi, u, v, c_u, c_v)$\;
            \Return{$\chi$}
        }
        Sample $i \sim [L]$ independently and u.a.r.\;
        $(\chi, u, v, c_u, c_v) \leftarrow \textsf{TruncatedVizing}(\chi, u,v,c_u, c_v, i+1)$\;
    }
    \caption{\textsf{ExtendColoring}$(\chi, u, v)$}
    \label{alg:extend}
\end{algorithm}

\medskip
\noindent
Using standard data structures, we can implement each iteration of \Cref{alg:extend} to run in time $\tilde O(L) \leq \tilde O(\Delta^2 + \sqrt{\Delta n})$. The following lemma, which we prove in \Cref{sec:anal}, shows that the algorithm terminates after $\tilde O(1)$ many iterations with constant probability.

\begin{lemma}\label{lem:extend works}
    \Cref{alg:extend} terminates within $\ell + 1$ iterations with probability at least $1/(160 \log n)$.
\end{lemma}

\subsection{Analysis}\label{sec:anal}

In order to analyse our algorithm, we define a (possibly infinite) rooted meta-tree $\mathcal T$ which captures information about all of the different possible executions of our algorithm. 

\subsubsection*{The Meta-Tree $\mathcal T$}

The meta-tree $\mathcal T$ is rooted at a meta-node $r$.
Every meta-node $x \in \mathcal T$ corresponds to an iteration of \Cref{alg:extend}, and the root-to-$x$ path from $r$ to $x$ in $\mathcal T$ corresponds to the execution of \Cref{alg:extend} leading to this iteration.
For each such meta-node $x$, we denote the state of an object $\Phi$ in \Cref{alg:extend} at the start of the corresponding iteration of the algorithm by $\Phi^{(x)}$, e.g.~$\chi^{(x)}$ denotes the coloring $\chi$ at the start of the corresponding iteration of the algorithm. Given a meta-node $x$, $P^{(x)}$ denotes the alternating path considered during the corresponding iteration of the algorithm. Consider the following definition:
\begin{wrapper}
    \begin{itemize}[leftmargin=1.5em]
    \item We say that $x$ is \emph{terminal} if the length of $P^{(x)}$ is at most $2L + 2$.
\end{itemize}
\end{wrapper}
Each meta-node $x$ has no children in $\mathcal T$ if it is terminal. Otherwise, $x$ has exactly $L$ children in $\mathcal T$, each one corresponding to a different choice of edge in $P^{(x)}$ for truncating the alternating path.\footnote{Intuitively, one can visualise constructing $\mathcal T$ by starting at the root $r$, where $\chi^{(r)} = \chi$ and $(u^{(r)}, v^{(r)}) = (u,v)$, and considering all possible executions of \Cref{alg:extend} for the different random choices made by the algorithm.}

\medskip
\noindent \textbf{Random Walks in $\mathcal T$:}
We first note that there is a one-to-one correspondence between root-to-leaf paths in $\mathcal T$ and executions of \Cref{alg:extend}. Furthermore, an execution of our algorithm defines a \emph{random walk} on $\mathcal T$ that starts at the root and, at each step, independently and uniformly samples a random child of the current meta-node (as long as the current meta-node is not terminal). Our algorithm successfully extends the coloring $\chi$ if and only if this random walk encounters a terminal meta-node. To this end, we show that such a random walk encounters a terminal meta-node within $O(\log n)$ steps with probability $\Omega (1/\log n)$, proving \Cref{lem:extend works}.

\medskip
\noindent \textbf{Classifications of Meta-Nodes:} 
Given a meta-node $x \in \mathcal T$, let $\tau^{(x)}$ denote the colors of the alternating path $P^{(x)}$ that we consider during this iteration. We refer to any pair of colors as a \emph{type}, and to $\tau^{(x)}$ as the type of $x$. Let $\tilde P^{(x)}$ denote the alternating path obtained by truncating the $\tau^{(x)}$-alternating path $P^{(x)}$ after the first $L+1$ edges.\footnote{Note that the alternating path $\tilde P^{(x)}$ is \emph{not} maximal, while the alternating path $P^{(x)}$ \emph{is} maximal.}
Let $r = x_0,\dots,x_i = x$ denote the root-to-$x$ path in $\mathcal T$. Consider the following definitions:

\begin{wrapper}
\begin{itemize}[leftmargin=1.5em]
    \item We say that $x$ is \emph{$\alpha$-dirty} if $\tau^{(x)} \cap (\tau^{(x_0)} \cup \dots \cup \tau^{(x_{i-2})}) \neq \emptyset$ and $\tau^{(x)} \cap \tau^{(x_{i-1})} = \emptyset$.
    \item We say that $x$ is \emph{$\beta$-dirty} if $\tau^{(x)} = \tau^{(x_{i-1})}$.
    \item We say that $x$ is \emph{$\alpha$-contaminated} (resp.~\emph{$\beta$-contaminated}) if it at least a $1/(10\ell)$ proportion of its children are \emph{$\alpha$-dirty} (resp.~\emph{$\beta$-dirty}).
    \item We say that $x$ is \emph{damaged} if it is \textbf{not} dirty and $P^{(x)}$ is {\bf not} a maximal $\tau^{(x)}$-alternating path in $\chi^{(r)}$ such that $\chi^{(r)}(e) = \chi^{(x)}(e)$ for all $e \in P^{(x)}$. 
\end{itemize}    
\end{wrapper}
Note that a meta-node $x$ cannot be both $\alpha$-dirty and $\beta$-dirty, but it can be both $\alpha$-contaminated and $\beta$-contaminated. We say that a meta-node $x$ is dirty (resp.~contaminated) if it is $\alpha$-dirty or $\beta$-dirty (resp.~$\alpha$-contaminated or $\beta$-contaminated).
Furthermore, it follows from \Cref{lm:overlap} that either $\tau^{(x)} = \tau^{(x_{i-1})}$ or $\tau^{(x)} \cap \tau^{(x_{i-1})} = \emptyset$, so $x$ is not dirty if and only if $\tau^{(x)} \cap (\tau^{(x_0)} \cup \dots \cup \tau^{(x_{i-1})}) = \emptyset$.

\subsubsection*{Properties of the Meta-Tree $\mathcal T$}

We now describe some properties of the meta-tree $\mathcal T$ that will be useful while analyzing the behaviour of random walks in $\mathcal T$.
We begin with the following lemmas which show that contaminated meta-nodes have lots of terminal children.

\begin{lemma}\label{lem:cont is good}
    Let $x$ be an $\alpha$-contaminated meta-node within the first $\ell$ levels of $\mathcal T$.
    Then at least a $1/(40\ell)$ proportion of the children of $x$ are terminal.
\end{lemma}

\begin{proof}
    Let $r = x_0,\dots,x_i=x$ denote the root-to-$x$ path in $\mathcal T$. Since $x$ is $\alpha$-contaminated, it has at least $L /(10 \ell)$ $\alpha$-dirty children $y_1, \dots,y_k$. Since $i \leq \ell$, at most $2\ell$ colors appear in $\tau^{(x_0)} \cup \dots \cup \tau^{(x_{i-2})}$. Thus, each of the $\alpha$-dirty children $y_1, \dots,y_k$ has one of at most $2\ell (\Delta+1) \leq 4\ell \Delta$ many types. Since $\tau^{(x)}$ is disjoint from the types of the $\alpha$-dirty children of $x$, we have that the alternating paths $P^{(y_1)},\dots,P^{(y_k)}$ corresponding to the children of $x$ are all maximal alternating paths with the same colors in $\chi^{(y_1)}$ and thus have a total length of at most $4\ell \Delta n$ (since the total length of the maximal alternating paths of any type is at most $n$).
    Furthermore, each alternating path in the collection $P^{(y_1)},\dots,P^{(y_k)}$ appears at most twice (once in each orientation), so we have at least $L/(20\ell)$ distinct alternating paths. Thus, by an averaging argument, at least half of these paths have length at most $ 8\ell \Delta n \cdot (20\ell/L) \leq L$, and hence correspond to terminal meta-nodes.
\end{proof}

\begin{lemma}\label{lem:cont is good 2}
    Let $x$ be a $\beta$-contaminated meta-node within the first $\ell$ levels of $\mathcal T$. 
    Then at least a $1/(20\ell)$ proportion of the children of $x$ are terminal.
\end{lemma}

\begin{proof}
    Since $x$ is $\beta$-contaminated, it has at least $L /(10 \ell)$ $\beta$-dirty children $y_1, \dots,y_k$. Let $\tau$ denote the type of $x$ (which is the same as the types of $y_1, \dots,y_k$). Consider the coloring $\chi^{(y_1)}$ and the uncolored edge $(u^{(y_1)}, v^{(y_1)})$. Let $P'$ (resp.~$P''$) denote the $\tau$-alternating path starting at $u^{(y_1)}$ (resp.~$v^{(y_1)}$), and let $w'$ (resp.~$w''$) denote it's other endpoint. Let $Q$ denote the path or cycle obtained by concatenating $P'$, $(u^{(y_1)}, v^{(y_1)})$, and $P''$. Then the colorings $\chi^{(y_1)},\dots,\chi^{(y_k)}$ only differ on the colors of the edges in $Q$. Furthermore, in each of these colorings, $Q$ always consists of either 1 or 2 $\tau$-alternating paths and an uncolored edge.

    Now, consider some meta-node $y_i$ and it's corresponding vertex $u^{(y_i)}$ in $Q$. There is an edge $f^{(y_i)}$ connecting $u^{(y_i)}$ to the first vertex on the $\tau$-alternating path $P^{(y_i)}$. We either have that the other endpoint of $f^{(y_i)}$ is $w'$ or $w''$, or that $P^{(y_i)}$ is vertex disjoint from $Q$. In the former case, we say that $y_i$ is a \emph{bad extension point}. Since there are only $2\Delta$ many edges incident on $w'$ and $w''$, there are at most $2\Delta$ bad extension points. Let $\{z_1,\dots,z_{k'}\} \subseteq \{y_1,\dots,y_k\}$ be the subset of the $y_i$ that are not bad extension points. Now, consider the collection of $\tau$-alternating paths $P^{(z_1)},\dots,P^{(z_{k'})}$. Since these are all vertex disjoint from $Q$, they are all maximal $\tau$-alternating paths with the same colors in $\chi^{(y_1)}$, and hence have a total length of at most $n$ without counting repeated paths. Furthermore, each alternating path in the collection $P^{(z_1)},\dots,P^{(z_{k'})}$ appears at most $2\Delta$ times since each endpoint of each path has maximum degree $\Delta$, and thus can only be incident on at most $\Delta$ edges $f^{(y_i)}$ connecting the path to some $u^{(y_i)}$. Thus, their total length (counting repeated paths) is at most $2\Delta n$. By an averaging argument, at least half of these paths have length at most
    $$\frac{4\Delta n}{k'} \leq 4\Delta n \cdot \frac{20\ell}{L}\leq L$$
    since $k' \geq L/(10 \ell) - 2\Delta \geq L/(20\ell)$. It follows that at least a $1/(20 \ell)$ proportion of the children of $x$ are terminal.
\end{proof}
The following lemma shows that meta-nodes cannot have many damaged children.

\begin{lemma}\label{lem:low damage}
    Let $x$ be a meta-node within the first $\ell$ levels of $\mathcal T$. Then at most a $1/(10\ell)$ proportion of the children of $x$ are damaged.
\end{lemma}

\begin{proof}
    We begin with the following structural claim about alternating paths.

    \begin{claim}
        Let $X \subseteq [\Delta + 1]$ a be set of at most $2\ell$ colors and let $\chi$ and $\chi'$ be $(\Delta + 1)$-edge colorings such that the set $D := \{e \in E \mid \{\chi(e), \chi'(e)\} \setminus X \neq \emptyset, \chi(e) \neq \chi'(e)\}$ has size at most $\Delta \ell$. For any $(\Delta + 1)$-edge coloring $\tilde \chi$, let $\mathcal P(\tilde \chi)$ denote the set of all maximal alternating paths w.r.t.~$\tilde \chi$ with colors in $[\Delta + 1] \setminus X$. Then we have that $|\mathcal P(\chi) \oplus \mathcal P(\chi')| \leq 6\Delta^2 \ell$.
    \end{claim}

    \begin{proof}
        Consider the following process: Start with the coloring $\chi_0 := \chi$, uncolor each of the edges in $D$ one by one to obtain a sequence of colorings $\chi_1,\dots, \chi_{|D|}$, then recolor each of the edges in $D$ with the color it receives under $\chi'$ to obtain a sequence of colorings $\chi_{|D|+1},\dots, \chi_{2|D|}$. Then we can observe that $\mathcal P(\chi') = \mathcal P(\chi_{2|D|})$ since any edge that receives a different color under $\chi'$ and $\chi_{2|D|}$ has colors in $X$ under both. Furthermore, for any $0 \leq i < 2|D|$, we have that $|\mathcal P(\chi_i) \oplus \mathcal P(\chi_{i+1})| \leq 3\Delta$. This is because each edge can belong to at most $\Delta$ different maximal alternating paths (one for each other color) and changing the color of an edge can cause 2 maximal alternating paths to merge into 1, and vice versa.
        It follows that
        $$|\mathcal P(\chi) \oplus \mathcal P(\chi')| \leq \sum_{i = 0}^{2|D|-1} |\mathcal P(\chi_i) \oplus \mathcal P(\chi_{i+1})| \leq 6\Delta |D| \leq 6\Delta^2 \ell.$$
    \end{proof}
    Let $x$ be a meta-node within the first $\ell$ levels of $\mathcal T$ with root-to-$x$ path $x_0,\dots,x_i$. If $x$ is terminal, then we are done. Otherwise, let $y_1,\dots,y_k$ denote the children of $x$ whose types are disjoint from the types $\tau^{(x_0)},\dots,\tau^{(x_i)}$ of its ancestors (note that any child of $x$ which is not in $\{y_1,\dots,y_k\}$ is dirty and hence not damaged). Let $e \in E$ be an edge such that $\chi^{(r)}(e) \notin \tau^{(x_0)} \cup \dots \cup \tau^{(x_i)}$ and $\chi^{(r)}(e) \neq \chi^{(y_j)}(e)$ for some $y_j$. Then we know that $e$ was involved in a fan construction in one of the iterations corresponding to the meta-nodes $x_0,\dots,x_i$. 
    Since there are at most $\Delta$ such edges per iteration, there are at most $\Delta \ell$ such edges in total. 
    We can observe that each alternating path in the collection $P^{(y_1)},\dots,P^{(y_k)}$ appears at most twice (once in each orientation) and that each of these paths receives the same colors and is a maximal alternating path in each of the colorings $\chi^{(y_1)},\dots,\chi^{(y_k)}$. Thus, applying the preceding claim with $X = \tau^{(x_0)} \cup \dots \cup\tau^{(x_i)}$, $\chi = \chi^{(r)}$ and $\chi' = \chi^{(y_1)}$, it follows that at most $6\Delta^2 \ell$ of the alternating paths in $P^{(y_1)},\dots,P^{(y_k)}$ are not maximal alternating paths with the same colors in $\chi^{(r)}$. Hence, at most $12\Delta^2 \ell$ of the meta-nodes $y_1,\dots,y_k$ are damaged. Since $12\Delta^2 \ell/L \leq 1/(10 \ell)$,  it follows that at most an $1/(10\ell)$ proportion of the children of $x$ are damaged.
\end{proof}

The following lemma is crucial for showing that there cannot be many long walks that do not contain any dirty or damaged meta-nodes.

\begin{lemma}\label{lem:expansion}
    Let $\tau_0,\dots,\tau_i$ be a sequence of types. Then there exist at most $n$ meta-nodes $x$ with root-to-$x$ path $x_0,\dots,x_i$ such that $\tau^{(x_j)} = \tau_j$ and $x_j$ is not dirty, damaged or terminal for all $x_j$.
\end{lemma}

\begin{proof}
    Let $\Gamma(\tau_0,\dots,\tau_i)$ denote the set of all such meta-nodes. We now prove the following claim.

    \begin{claim}
        For any distinct meta-nodes $x,y \in \Gamma(\tau_0,\dots,\tau_i)$, the $\tau_i$-alternating paths $\tilde P^{(x)}$ and $\tilde P^{(y)}$ are vertex disjoint.
    \end{claim}

    \begin{proof}
        We prove this by induction. For any type $\tau_0$ we have that $\Gamma(\tau_0) \subseteq \{x_0\}$. Thus, this claim holds trivially for $i=0$. For the inductive step, suppose that the claim holds for some $0 \leq j-1 \leq i-1$. If the type $\tau_j$ shares a color with any of the types $\tau_0,\dots,\tau_{j-1}$, then $\Gamma(\tau_0,\dots,\tau_j) = \emptyset$ since the types of the meta-nodes on the root-to-$x$ path of a non-dirty meta-node $x$ must be disjoint.
        Thus, we can assume that the types $\tau_0,\dots,\tau_{j}$ are disjoint. Now, let $x,y \in \Gamma(\tau_0,\dots,\tau_j)$ and let $(u^{(x)},v^{(x)})$ and $(u^{(y)},v^{(y)})$ be the uncolored edges in the iterations of our algorithm corresponding to meta-nodes $x$ and $y$ respectively. By the induction hypothesis, the vertices $u^{(x)}$ and $u^{(y)}$ are distinct since they either lie on different positions of the same $\tau_{j-1}$-alternating path or on different vertex disjoint paths.
        The $\tau_j$-alternating paths $P^{(x)}$ and $P^{(y)}$ constructed by our algorithm have $u^{(x)}$ and $u^{(y)}$ as endpoints respectively. Since $x$ and $y$ are not dirty and not damaged, these alternating paths are also maximal $\tau_j$-alternating paths with the exact same colors under $\chi^{(r)}$, so they are either distinct paths or the same path but with different orientations. In the first case, $\tilde P^{(x)}$ and $\tilde P^{(y)}$ are clearly disjoint. In the second case, we note that $x$ and $y$ are not terminal; hence, $P^{(x)}$ has length at least $2L + 3$, so $\tilde P^{(x)}$ and $\tilde P^{(y)}$ are also vertex disjoint since they have length at most $L+1$.
    \end{proof}

    Thus, the size of $\Gamma(\tau_0,\dots,\tau_i)$ is upper bounded by the maximum size of a collection of vertex disjoint paths, which is at most $n$.
\end{proof}

\subsubsection*{Analyzing the Random Walk}

Consider the following definition of a \emph{good} walk.

\begin{wrapper}
\begin{itemize}[leftmargin=1.5em]
    \item We say that a walk $x_0, x_1,\dots$ in the meta-tree $\mathcal T$ is \emph{good} if it encounters either a terminal or contaminated meta-node within the first $\ell$ steps.
\end{itemize}    
\end{wrapper}
The following lemma bounds the probability that a good random walk has length at most $\ell + 1$.

\begin{lemma}\label{lem:good walks are good}
    Let $x_0,x_1,\dots$ be a random walk in $\mathcal T$. Given that the random walk is good, it has length at most $\ell + 1$ with probability $1/(40\ell)$.
\end{lemma}

\begin{proof}
    Since the random walk is good, we know that $x_i$ is terminal or contaminated for some $i$.
    If $x_i$ is terminal, then the random walk has length $i \leq \ell$. 
    If $x_i$ is $\alpha$-contaminated, then it follows from \Cref{lem:cont is good} that $x_{i+1}$ is terminal (and hence that the random walk has length at most $i+1$) with probability at least $1/(40\ell)$. Similarly, if $x_i$ is $\beta$-contaminated, then it follows from \Cref{lem:cont is good 2} that the random walk has length at most $i+1$ with probability at least $1/(20\ell)$.
\end{proof}

It follows from \Cref{lem:good walks are good} that it suffices to lower bound the probability that a random walk in the meta-tree $\mathcal T$ is good. More precisely, we have the following.
Let $\mathcal W_{\mathcal T}$ denote the collection of all walks in $\mathcal T$ and let $(x_i)_i \sim \mathcal W_{\mathcal T}$ denote a random walk. Then we have that

\begin{align}
    \Pr_{(x_i)_i \sim \mathcal W_{\mathcal T}}[ |(x_i)_i| \leq \ell + 1 ] &\geq \Pr_{(x_i)_i \sim \mathcal W_{\mathcal T}}[ |(x_i)_i| \leq \ell + 1 \mid (x_i)_i \textrm{ is good} ] \cdot \Pr_{(x_i)_i \sim \mathcal W_{\mathcal T}}[ (x_i)_i \textrm{ is good} ]\nonumber\\
    &\geq \frac{1}{40\ell} \cdot \Pr_{(x_i)_i \sim \mathcal W_{\mathcal T}}[ (x_i)_i \textrm{ is good}].\label{eq:random:1}
\end{align}
Here, the first inequality follows from conditioning on the event that the random walk $(x_i)_i$ is good, and the second inequality follows from \Cref{lem:good walks are good}.

\medskip
\noindent \textbf{The Meta-Subtree $\mathcal T'$:}
We define a meta-subtree $\mathcal T'$ of $\mathcal T$ as follows: Start with the meta-tree $\mathcal T$ and remove all of the (strict) descendants of contaminated meta-nodes. 
We can observe that the probability of a random walk in $\mathcal T$ being good is the same as the probability of a random walk in $\mathcal T'$ having length at most $\ell$, i.e.~that
\begin{equation}\label{eq:random:2}
\Pr_{(x_i)_i \sim \mathcal W_{\mathcal T}}[ (x_i)_i \textrm{ is good} ] = \Pr_{(x_i)_i \sim \mathcal W_{\mathcal T'}}[ |(x_i)_i| \leq \ell ].     
\end{equation}
To see why this is true, consider a mapping $\phi : \mathcal W_{\mathcal T} \longrightarrow \mathcal W_{\mathcal T'}$ which maps a walk $(x_i)_i$ in $\mathcal T$ to a walk $\phi((x_i)_i)$ which is defined as the prefix of $(x_i)_i$ which is contained in $\mathcal T'$. Then we can observe that $(x_i)_i$ is good if and only if $|\phi((x_i)_i)| \leq \ell$.

\medskip
\noindent \textbf{The Meta-Subtree $\mathcal T''$:} We define a meta-subtree $\mathcal T''$ of $\mathcal T'$ as follows: Start with the meta-tree $\mathcal T'$ and remove all of the meta-nodes that are dirty or damaged (and their descendants) from $\mathcal T'$. It turns out that bounding the length of random walks in $\mathcal T''$ is much easier than bounding the length of random walks in $\mathcal T'$. Thus, we only want to consider random walks in $\mathcal T'$ which are also contained in $\mathcal T''$ for the first $\ell$ steps. In particular, we use the following bound
\begin{equation}\label{eq:random:3}
    \Pr_{(x_i)_i \sim \mathcal W_{\mathcal T'}}[ |(x_i)_i| \leq \ell ] \geq \Pr_{(x_i)_i \sim \mathcal W_{\mathcal T''}}[ |(x_i)_i| \leq \ell ] \cdot \Pr_{(x_i)_i \sim \mathcal W_{\mathcal T'}}[ (x_i)_{i \leq \ell} \subseteq \mathcal T'' ].
\end{equation}
Given any meta-node $x \in \mathcal T'$ that has children, we know that at most a $1/(5\ell)$ proportion of the children of $x$ are dirty since $x$ is not contaminated.\footnote{Recall that since $x$ is neither $\alpha$ nor $\beta$-contaminated, at most a $1/(10\ell)$ proportion of its children at $\alpha$-dirty and at most a $1/(10\ell)$ proportion of its children at $\beta$-dirty} 
Furthermore, it follows from \Cref{lem:low damage} that at most a $1/(10\ell)$ proportion of the children of $x$ are damaged.
It follows that a random walk in $\mathcal T'$ does not encounter any dirty or damaged meta-nodes within the first $\ell$ steps with probability at least $(1 - 3/(10\ell))^\ell \geq 1 - 3/10 = 7/10$. Thus, we have that
\begin{equation}\label{eq:random:4}
    \Pr_{(x_i)_i \sim \mathcal W_{\mathcal T'}}[ (x_i)_{i \leq \ell} \subseteq \mathcal T'' ] \geq \frac{7}{10}.
\end{equation}
Combining these inequalities, it follows that
\begin{equation}\label{eq:reduce to clean}
    \Pr_{(x_i)_i \sim \mathcal W_{\mathcal T}}[ |(x_i)_i| \leq \ell + 1 ] \geq \frac{1}{80\ell} \cdot \Pr_{(x_i)_i \sim \mathcal W_{\mathcal T''}}[ |(x_i)_i| \leq \ell ].
\end{equation}
We now lower bound the probability that a random walk in $\mathcal T''$ terminates within at most $\ell$ steps.

\medskip
\noindent \textbf{Random Walks in $\mathcal T''$:} Let $\mathcal T''_{\ell+1}$ denote the meta-nodes at depth $\ell+1$ in $\mathcal T''$. Given a walk $(x_i)_i$ in $\mathcal T''$, $(x_i)_i$ has length greater than $\ell$ if and only if $(x_i)_i$ contains some meta-node from $\mathcal T''_{\ell+1}$, thus
$$ \Pr_{(x_i)_i \sim \mathcal W_{\mathcal T''}}[| (x_i)_i | > \ell ] = \Pr_{(x_i)_i \sim \mathcal W_{\mathcal T''}}[\mathcal T''_{\ell+1} \cap (x_i)_i \neq \emptyset]. $$

\begin{lemma}\label{lemma:node mass}
    For any $x \in \mathcal T''_{\ell+1}$, we have that $\Pr_{(x_i)_i \sim \mathcal W_{\mathcal T''}}[ x \in (x_i)_i] \leq (2/L)^{\ell+1}$.
\end{lemma}

\begin{proof}
    Give some non-leaf meta-node $y \in \mathcal T''$, we can see that $y$ has at least $L \cdot (1 - 3/(10 \ell)) \geq L/2$ children in $\mathcal T''$. This is because $y$ is not contaminated (otherwise it would be a leaf) and hence has at most $L/(5\ell)$ dirty children in $\mathcal T'$ and at most $L/(10\ell)$ damaged children in $\mathcal T'$. Thus, the probability that a random walk in $\mathcal T''$ contains a child $y'$ of $y$ given that it contains $y$ is at most $1/(L/2)$. It follows that the probability that a random walk in $\mathcal T''$ contains a meta-node $x \in \mathcal T''_{\ell+1}$ is at most $(2/L)^{\ell + 1}$.
\end{proof}

\begin{lemma}\label{lemma:small layer}
    $|\mathcal T''_{\ell+1}| \leq \Delta^{2\ell} nL$.
\end{lemma}

\begin{proof}
    Recall the definition of $\Gamma$ from the proof of \Cref{lem:expansion}. Let $x \in \mathcal T''_{\ell+1}$ and let $x_0,\dots,x_{\ell + 1}$ denote the root-to-$x$ path in $\mathcal T$. Then it must be the case that $x_{\ell + 1}$ is the child of some meta-node in $\Gamma(\tau^{(x_0)},\dots,\tau^{(x_{\ell})})$ since it's parent cannot be terminal. It follows that every meta-node in $\mathcal T''_{\ell+1}$ is a child of some meta-node in
    $$\bigcup_{\tau \in \binom{[\Delta+1]}{2}^{\ell}} \Gamma(\tau).$$
    Since any meta-node has at most $L$ children and $|\Gamma(\tau)| \leq n$ by \Cref{lem:expansion}, it follows that
    $$ |\mathcal T''_{\ell+1}| \leq L \cdot \sum_{\tau \in \binom{[\Delta+1]}{2}^{\ell}} |\Gamma(\tau)| \leq  L \cdot (\Delta^2)^\ell \cdot n \leq \Delta^{2\ell} nL.$$
\end{proof}
It follows from Lemmas~\ref{lemma:node mass} and \ref{lemma:small layer} that
$$ \Pr_{(x_i)_i \sim \mathcal W_{\mathcal T''}}[\mathcal T''_{\ell+1} \cap (x_i)_i \neq \emptyset] \leq \sum_{x \in \mathcal T''_{\ell+1}} \Pr_{(x_i)_i \sim \mathcal W_{\mathcal T''}}[ x \in (x_i)_i] \leq |\mathcal T''_{\ell+1}| \cdot \left(\frac{2}{L}\right)^{\ell+1} \leq nL \cdot \left( \frac{2\Delta^2}{L} \right)^{\ell+1} \leq \frac{1}{2}. $$
Hence, we have that
\begin{equation}\label{eq:random:5}
    \Pr_{(x_i)_i \sim \mathcal W_{\mathcal T''}}[| (x_i)_i | \leq \ell ] = 1 - \Pr_{(x_i)_i \sim \mathcal W_{\mathcal T''}}[| (x_i)_i | > \ell ] \geq \frac{1}{2}.
\end{equation}
Combining Equations~(\ref{eq:reduce to clean}) and (\ref{eq:random:5}), it follows that
\begin{equation}
    \Pr_{(x_i)_i \sim \mathcal W_{\mathcal T}}[ |(x_i)_i| \leq \ell + 1 ] \geq \frac{1}{160\ell}.
\end{equation}
In other words, a random walk in the meta-tree $\mathcal T$ terminates within $\ell + 1$ steps with probability at least $1/(80\ell)$.

\subsection{Proof of \Cref{thm:extension expected}}

Suppose that we run \Cref{alg:extend} for $\kappa = (\ell + 1) \cdot 1600 \log^2 n$ iterations. It follows from \Cref{lem:extend works} that, given an arbitrary input, the probability of the algorithm successfully extending the coloring within $\ell + 1$ iterations is at least $1/(160\log n)$. Thus, we can split these $\kappa$ updates into batches of $\ell + 1$ iterations, where the probability that the algorithm successfully extends $\chi$ within some batch is at least $1/(160\log n)$. Thus, the probability that the algorithm fails to extend the coloring within $\kappa$ iterations it at most
$$ \left(1 - \frac{1}{160 \log n} \right)^{1600 \log^2 n} \leq \exp \!\left( - \frac{1600 \log^2 n}{160 \log n} \right) = \frac{1}{n^{10}}. $$
In the low probability event that the algorithm does not successfully extend the coloring within $\kappa$ iterations, then we can extend it in $O(n)$ time using Vizing's original algorithm. Thus, our algorithm runs in time $\kappa \cdot \tilde O(L) \leq \tilde O(\Delta^2 + \sqrt{\Delta n})$ both with high probability and in expectation.

%% file: stars.tex
\section{Color Extension to Stars}\label{sec:stars}
\label{sec:star}

The idea of {\em coloring} stars was first proposed in \cite{BhattacharyaCCSZ24}. In this section, we provide a simpler algorithm with improved runtime bound.
\begin{lemma}[\Cref{lem:key 1} Restated]
There is an algorithm $\textnormal{\texttt{ColorLightStars}}$ that, given a graph $G$, a partial $(\Delta + 1)$-edge coloring $\chi$ of $G$ and subset $U^\star \subseteq V$ such that:
\begin{itemize}
    \item $|\miss_\chi(u)| \geq d$ for all $u \in U^\star$ for some positive integer $d \in \mathbb N$,
    \item there are $\lambda$ uncolored edges incident on $U^\star$,
\end{itemize}
extends the coloring $\chi$ to $\Omega(\lambda)$ uncolored edges incident on $U^\star$ in $\tilde O(\lambda \Delta + \Delta m / d)$ expected time.
\end{lemma}

\subsection{Recap on Vizing's algorithm}\label{prelim:vizing} Our algorithm will again use the original Vizing's algorithm in \cite{Vizing} that extends any partial edge coloring $\chi$ by one more colored edge. Since we will not need the involved modifications as in previous sections and only use the basic version in a slightly different way, for reader's convenience, we will describe it again here with simpler terms.

Let $(u, v)\in E$ be any uncolored edge under $\chi$; that is, $\chi(u, v) = \bot$. For any vertex $w\in V$, specify an arbitrary color $c_\chi(w)\in \miss_\chi(w)$, and also specify an arbitrary color $z\in \miss_\chi(v)$ different from $c_\chi(v)$. Then find a sequence of distinct neighbors $v = v_0, v_1,v_2, \ldots, v_k$ of $u$ such that the following holds; this sequence $v_1,v_2, \ldots, v_k$ is usually called a {\em Vizing fan}.
\begin{itemize}[leftmargin=*]
	\item For any $2\leq i\leq k$, $\chi(u, v_i) = c_\chi(v_{i-1})$, and $z = \chi(u, v_1)$.
	\item Either $c_\chi(v_k)\in \miss_\chi(u)$, or there exists an index $1\leq j<k$ such that $\chi(u, v_j) = c_\chi(v_k)$.
\end{itemize}

If $c_\chi(v_k)\in \miss_\chi(u)$, then rotate the coloring around $u$ as: $\chi(u, v_i)\leftarrow \chi(u, v_{i+1}), 0\leq i<k$, and $\chi(u, v_k)\leftarrow c_\chi(v_k)$.

Now, let us assume $c_\chi(v_k)\notin \miss_\chi(u)$, so there must exist an index $1\leq j<k$ such that $\chi(u, v_j)=c_\chi(v_k)$. Take an arbitrary color $x\in \miss_\chi(u)$, and define $y = \chi(u, v_j)$. Let $P$ be the maximal $\{x, y\}$-alternating path beginning at $u$.
\begin{enumerate}[(1),leftmargin=*]
	\item $P$ does not end at $v_{j-1}$.
	
	Apply a rotation operation: $\chi(u, v_i)\leftarrow \chi(u, v_{i+1}), 0\leq i<j$, and flip the maximal $\{x, y\}$-alternating path $P$. Finally, assign $\chi(u, v_{j})\leftarrow x$.
	
	\item $P$ ends at $v_{j-1}$.
	
	Flip the color of the maximal $\{x, y\}$-alternating path from $u$, then apply a rotation operation: $\chi(u, v_i)\leftarrow \chi(u, v_{i+1}), 0\leq i<k$, then assign color $\chi(u, v_k)\leftarrow y$.
\end{enumerate}

The runtime of Vizing's procedure is bounded by $\tilde{O}(\Delta + |P|)$.

\subsection{Algorithm Description}
Throughout our color extension procedure, we will be repeatedly assigning proper colors to edges in $S$ defined below while decreasing the size of some sets $\miss_\chi(u), u\in U^\star$. Therefore, we will keep a terminal subset $W\subseteq U^\star$ with the property that for any $u\in W$, we have $|\miss_\chi(u)|\geq d/2$. Initially, set $W\leftarrow U^\star$.

\begin{framed}
	\noindent We will maintain the following data structures.
	\begin{itemize}[leftmargin=*]
		\item A set $S$ of currently uncolored edges incident on $W$.
At the beginning we choose $S$ to be those $\lambda$ uncolored edges incident on $W$, which is initially $U^\star$.
While initially $S$ is the set of uncolored edges incident on $U^\star$, throughout the execution of the algorithm some uncolored edges may be removed from $S$; however, as we will argue (refer to \Cref{wrapup}), whenever a constant fraction of edges have been removed from $S$, it implies that a constant fraction of the original set of $\lambda$ uncolored edges have been colored.

   We assign all edges in $S$ a direction:
for each edge $(v, u)\in S$ such that $u\in W$, orient this edge from $v$ to $u$; if both $u, v$ are in $W$, orient this edge arbitrarily.
		
		\item For each vertex $v\in V$, our algorithm will explicitly maintain a color $c_\chi(v)\in \miss_\chi(v)$.
		
		\item For each vertex $u\in W$ and each of its neighbor $v$ such that $\chi(u, v) = \bot$, maintain a tentative color $\clr_\chi(v, u)\in \miss_\chi(v)$ with the following two properties.
		\begin{itemize}[leftmargin=*]
			\item $\clr_\chi(v, u)\neq c_\chi(v)$ for all $u\in W$  such that $\chi(u, v) = \bot$;
			\item $\clr_\chi(v, u)\neq \clr_\chi(v, u')$ for all pairs of distinct neighbors $u, u'$ of $v$ in $W$ such that $\chi(u, v) = \chi(u',v) = \bot$. 
		\end{itemize}
            Note that such choices of $\clr_\chi(\cdot, \cdot)$ which are compatible with $c_\chi(\cdot)$ always exist since the palette size is $\Delta+1$.
  
		\item For each color $x$, maintain a list $\lst_\chi(u, x) = \{(v, u)\in S\mid u \in W, x = \clr_\chi(v, u)\}$.
	\end{itemize}
\end{framed}

We first argue that these data structures can be maintained efficiently.
\begin{lemma}\label{clm:ds-maintain}
    Suppose the partial coloring $\chi$ and the set $S$ has undergone $k$ changes. Then, the data structures $\miss_\chi(\cdot), c_\chi(\cdot), \clr_\chi(\cdot, \cdot)$ can be maintained in $\tilde{O}(k)$ time.
\end{lemma}
\begin{proof}
    For each vertex $v\in V$, maintain a list of colors $\clr_\chi(v) = \{\clr_\chi(v, u)\mid (v, u)\in S\}$, and also maintain the set $\miss_\chi(v)\setminus \clr_\chi(v)$. To recover the data structures after $k$ changes to $\chi$ and $S$, consider the following two steps.
    \begin{enumerate}[(1),leftmargin=*]
        \item Temporarily assign $c_\chi(v)\leftarrow\bot$. Initialize a color set $C_v\leftarrow \miss_\chi(v)\setminus \clr_\chi(v)$ using our old version of the data structure. Also, let $S_v$ be the set of new edges $(v, u)\in S$ added to $S$, and temporarily assign $\clr_\chi(v, u)\leftarrow\bot$ for the moment.
        
        Then, for each old color $x$ removed around $v$ in the new $\chi$, add $x$ to $C_v$. For each new color $x$ added around $v$, remove $x$ from $C_v$, and also if there is currently an edge $(v, u)\in S$ such that $\clr_\chi(v, u) = x$, temporarilly assign $\clr_\chi(v, u) \leftarrow\bot$ and add $(v, u)$ to $S_v$. Each of these steps can be implemented in $O(\log\Delta)$ time.

        \item We can see that $S_v$ is currently the set of edges $(v, u)\in S$ such that $\clr_\chi(v, u) = \bot$, and $C_v = \miss_\chi(v)\setminus \clr_\chi(v)$. Since there are $\Delta+1$ colors, we know that $|C_v|\geq |S_v|+1$. Therefore, we can allocate $C_v$ properly to $\clr_\chi(v, u), (v, u)\in S_v$ as well as $c_\chi(v)$ without any conflicts.
    \end{enumerate}
    The total runtime of the above two steps is $\tilde{O}(k)$.
\end{proof}

For each vertex $u\in U^\star$, define a directed graph $F_\chi(u)$ in the following way; this digraph $F_\chi(u)$ will not be maintained by our algorithm, but only computed on-the-fly when needed.
\begin{itemize}[leftmargin=*]
	\item \textbf{Vertices.} The vertex set of digraph $F_\chi(u)$ is the set of all colored neighbors 
 $N_\chi(u)$ of $u$.
	\item \textbf{Edges.} For each $z, w\in N_\chi(u)$, if $\chi(u, w) = c_\chi(z)$, then add a directed edge $(z, w)$ to $F_\chi(u)$.
\end{itemize}
So by definition, the out-degree of each vertex in $F_\chi(u)$ is at most one, and so $F_\chi(u)$ is a directed pseudo-forest; in particular, note that any weakly connected component in $F_\chi(u)$ is a directed tree plus at most one extra edge.

\begin{definition}[uncolored edge classification]\label{def-social}
	For any directed edge $(v, u)\in S$ such that $u\in W$, the directed edge $(v, u)$ is classified into one of the following four types.
	\begin{itemize}[leftmargin=*]
		\item If $\clr_\chi(v, u)\in \miss_\chi(u)$, then the directed edge $(v, u)$ is called \textbf{ready}.
		
		\item Otherwise, if there exists another uncolored directed edge $(v', u)\in S$ such that $\clr_\chi(v\rightarrow u) = \clr_\chi(v'\rightarrow u)$, then the directed edge $(v, u)$ is called \textbf{social}.
		
		\item Otherwise, let $w\in N_\chi(u)$ be the unique vertex such that $\chi(u, w) = \clr_\chi(v, u)$. Let $T_\chi^w(u)\subseteq F_\chi(u)$ be the weakly connected component that contains vertex $w$. If $w$ is the only vertex $w'\in T_\chi^w(u)$ such that $|\lst_\chi(u, \chi(u, w'))| = 1$, then the directed edge $(v, u)$ is called \textbf{independent}.
		
		\item Otherwise, the directed edge $(v, u)$ is called \textbf{lonely}.
	\end{itemize} 
	
\end{definition}

Our main color extension algorithm is described as follows.

\begin{framed}
	\noindent Define a length parameter $L = \frac{200\Delta m}{d\lambda}$. The algorithm consists of iterations of random color extension, which are repeated as long as $|S|\geq \frac{3\lambda}{4}$.  We start by uniformly sampling an uncolored edge $(v, u)\in S$ incident on $W$, and then sampling uniformly at random a color $x\in \miss_\chi(u)$.
The execution of an iteration splits into the following four cases.

	\begin{enumerate}[(1),leftmargin=*]
		\item \label{ext:direct} If $(v, u)$ is classified as ready, then assign $\chi(v, u) \leftarrow \clr_\chi(v, u)$. In this case, this iteration of color extension would be considered {\em successful}.
		
		\item \label{ext:social}Otherwise, check if $(v, u)$ is social, which can be done by checking if $|\lst_\chi(u, \clr_\chi(v, u))|>1$. If so, let $P$ be the maximal $\{x, y\}$-alternating path starting from vertex $v$ where $y = \clr_\chi(v, u)$.
		
		If $|P|\leq L$ and $P$ does not terminate at vertex $u$, then flip this maximal $\{x, y\}$-alternating path $P$, and assign $\chi(v, u) \leftarrow x$. In this case, this iteration of color extension would be considered successful.
		
		\item  \label{ext:vizing} Otherwise, construct the entire graph $F_\chi(u)$ and check if $(v, u)$ is independent. If so, try to apply Vizing's procedure described in \Cref{prelim:vizing} to extend $\chi$ to $(v, u)$, where vertex $v$ uses the missing color $\clr_\chi(v, u)$, vertices $w$ use missing colors $c_\chi(w), \forall w\in N_\chi(u)$, and $u$ uses the missing color $x\in \miss_\chi(u)$. Let $Q$ be the maximal $\{x, z\}$-alternating path which is the Vizing chain for $(v, u)$; in Vizing's procedure, $Q$ could be an empty path.
		
		If $|Q|\leq L$, then complete the Vizing procedure, which extends $\chi$ to $(v, u)$. In this case, this iteration of color extension would be considered successful. 
		
		\item \label{ext:pairing} Otherwise, suppose $\chi(u, w) = \clr_\chi(v, u)$ and let $T_\chi^w(u)\subseteq F_\chi(u)$  be the weakly connected component containing $w$. Since $(v, u)$ is not independent, by \Cref{def-social}, there exists another edge $(v', u)\in S$ and $w'\in T_\chi^w(u)$, such that $\clr_\chi(v', u) = \chi(u, w')$ and $|\lst_\chi(u, \chi(u, w'))|=1$. This implies that $(v', u)\in S$ is also lonely.
		
		Let $T\subseteq T_\chi^w(u)$ be a spanning tree such that all edges are in the direction from children to parents in $T_\chi^w(u)$. Let $t\in T$ be the least common ancestor of $w, w'$ in tree $T$, and define $w_0(=v), w_1(=w), \ldots, w_l = t$ and $w'_0(=v'), w'_1(=w'), \ldots, w'_{k} = t$ be the tree paths from $w, w'$ to $t$. As $t$ is the least common ancestor, these two tree paths are internally disjoint.
		
		Perform the following color shifts to $\chi$:  $$\chi(u, w_i)\leftarrow \chi(u, w_{i+1}), \forall \; 0\leq i < l-1$$
		$$\chi(u, w'_j)\leftarrow \chi(u, w'_{j+1}), \forall \; 0\leq j\leq k-1$$
		Then, uncolor the edges $\chi(u, w_{l-1}), \chi(u, w'_{k-1})\leftarrow\bot$. After that, update the edge set $$S\leftarrow S\cup \{(w_{l-1}, u), (w'_{k-1}, u)\}\setminus \{(v, u), (v', u)\}$$
		Finally, assign tentative colors $$\clr_\chi(w_{l-1}, u), \clr_\chi(w'_{k-1}, u) \leftarrow \chi(u, t)$$
		\end{enumerate}
  After the above steps (in all four cases), we maintain the data structures regarding $\clr_\chi(\cdot, \cdot), c_\chi(\cdot), \lst_\chi(\cdot, \cdot)$ in the straightforward manner. Finally, if the size of some set $\miss_\chi(u_0), u_0\in W$ drops below $d/2$, remove $u_0$ from $W$, and also remove all edges $(v, u_0)$ from $S$.

\end{framed}

\subsection{Runtime Analysis}
\begin{lemma}\label{sum-alt-path}
	Given any partial edge coloring $\chi$ of $G = (V, E)$ and any color $x\in \{1, 2, \ldots, \Delta+1\}$, for any $y\in \{1, 2, \ldots, \Delta+1\}\setminus \{x\}$, let $L_{x, y}$ denote the total length of all maximal $\{x, y\}$-alternating paths of lengths at least $2$. Then $\sum_{y, y\neq x}L_{x, y} < 3m$.
\end{lemma}
\begin{proof}
	For any maximal $\{x, y\}$-alternating path $P$ such that $|P|\geq 2$, the number of edges in $P$ with color $x$ is at most twice the number of edges in $P$ with color $y$. Therefore, taking a summing over all $y$ such that $y\neq x$ and all such maximal $\{x, y\}$-alternating paths, we have 
	$$\sum_{y, y\neq x}L_{x, y} \leq 3\cdot |\{e\in E\mid \chi(e)\neq x\}| < 3m.$$
\end{proof}

\begin{lemma} \label{runtime}
    The runtime of each iteration of color extension is $\tilde{O}(\Delta + L)$.
\end{lemma}
\begin{proof}
    First, let us analyze the runtime in each of the four cases.
    \begin{itemize}[leftmargin=*]
        \item In Step \ref{ext:direct}, the runtime is $O(1)$.

        \item In Step \ref{ext:social}, checking if an edge $(v, u)\in S$ is social takes $O(1)$ time. Checking if the length of the maximal alternating path $P$ exceeds $L$ takes $\tilde{O}(L)$ time by tracing this path $P$ until it ends or reaches length $L+1$. Flipping this path also takes $\tilde{O}(L)$ time.

        \item In Step \ref{ext:vizing}, constructing the entire digraph $F_\chi(u)$ takes $\tilde{O}(\Delta)$ time. Deciding if $(v, u)$ is independent also takes $\tilde{O}(\Delta)$ time by going over all vertices $w'$ in the weakly connected component $T_\chi^w(u)\subseteq F_\chi(u)$ containing the unique vertex $w\in N_\chi(u)$ such that $\chi(u, w) = \clr_\chi(v, u)$ and checking if $|\lst_\chi(u, \chi(u, w'))|=1$. After that, running Vizing's procedure takes time $\tilde{O}(\Delta+L)$.

        \item In Step \ref{ext:pairing}, finding the other edge $(v', u)\in F$ and $w'\in T_\chi^w(u)$ takes $\tilde{O}(\Delta)$ time. Later on, finding the two tree paths and shifting the colors also take time $\tilde{O}(\Delta)$.

    \end{itemize}
        In each of the four cases,  to recover the validity of the data structures $\clr_\chi(\cdot, \cdot), c_\chi(\cdot), \lst_\chi(\cdot, \cdot)$, it takes $\tilde{O}(1)$ for each color change to $\chi$. As we have just proved that the total number of changes to $\chi$ and $S$ is bounded by $O(\Delta+L)$, the runtime to maintain $\clr_\chi(\cdot, \cdot), c_\chi(\cdot), \lst_\chi(\cdot, \cdot)$ would be $\tilde{O}(\Delta+L)$, according to \Cref{clm:ds-maintain}. Finally, since each iteration adds one colored edge, $W$ loses at most $2$ vertices in each iteration, which takes runtime $O(\Delta)$ to update $S$.
\end{proof}

So it suffices to analyze the total number of iterations. The following claim serves as the basis of our analysis.
\begin{lemma} \label{base}
	During our algorithm, whenever we have $\clr_\chi(v_1, u_0) = \clr_\chi(v_2, u_0)$ for two different edges $(v_1, u_0), (v_2, u_0)\in S$, meaning that both $(v_1, u_0), (v_2, u_0)$ are social, they will stay in $S$ as two social edges until one of the following events occur.
	\begin{enumerate}[leftmargin=*]
		\item Vertex $v_1$ or $v_2$ is incident on a newly colored edge $(v, u)$ with new color $\chi(v, u) \leftarrow \clr_\chi(v_1, u_0)$.
		\item $v_1$ or $v_2$ becomes an endpoint of a flipped maximal alternating path $P$ or $Q$ computed in Step \ref{ext:social} or Step \ref{ext:vizing}, which successfully extends $\chi$ to one more uncolored edge.
	\end{enumerate}
\end{lemma}
\begin{proof}
	If there are no successful color extensions since we have the equality $\clr_\chi(v_1, u_0) = \clr_\chi(v_2, u_0)$, then the partial coloring $\chi$ could only be modified by Step \ref{ext:pairing}. Since both $(v_1, u_0)$ and $(v_2, u_0)$ are social, by the description of Step \ref{ext:pairing}, $(v_1, u_0)$ and $(v_2, u_0)$ will stay in $S$ after this iteration of Step \ref{ext:pairing}. Furthermore, during Step \ref{ext:pairing}, the new color added around vertex $v_b, \forall b\in \{1, 2\}$, can only be $c_\chi(v_b)$ or $\clr_\chi(v_b, u'), u'\neq u$. Since our data structure has guaranteed that $c_\chi(v_b)\neq \clr_\chi(v_b, u_0)$ and $\clr_\chi(v_b, u')\neq \clr_\chi(v_b, u_0), u_0'\neq u_0$, both tentative colors $\clr_\chi(v_1, u_0)$ and $\clr_\chi(v_2, u_0)$ will not change, and so $(v_1, u_0)$ and $(v_2, u_0)$ will remain social.
	
	We may henceforth assume that there has been a successful color extension. Further, we assume that neither $v_1$ nor $v_2$ is incident on the newly colored edge $(v, u)$ such that $\chi(v, u)$ is assigned $\clr_\chi(v_1, u_0)$, otherwise event 1 occurs and we are done. Consequently, the sets $\miss_\chi(v_b), b\in \forall \{1, 2\}$, could change only in one of the following two cases.
	\begin{itemize}[leftmargin=*]
	\item $v_b$ is an endpoint of the maximal alternating path $P$ or $Q$ in Step \ref{ext:social} \ref{ext:vizing}. Thus, event 2 occurs.

	\item $v_b$ lies on the Vizing fan in Step \ref{ext:vizing}.
	
In this case, since we assumed that event 1 does not occur and so $(v_b, u_0)$ is not colored right away, the new possible color added around vertex $v_b$ can only be $c_\chi(v_b)$ or $\clr_\chi(v_b, u_0')$ for some $u_0'\neq u_0$. Since our data structure has guaranteed that $c_\chi(v_b)\neq \clr_\chi(v_b, u_0)$ and $\clr_\chi(v_b, u')\neq \clr_\chi(v_b, u_0), u_0'\neq u_0$, the tentative color $\clr_\chi(v_b, u_0)$ does not change. Hence, $(v_1, u_0), (v_2, u_0)$ will remain social.
\end{itemize}	
\end{proof}

\begin{corollary}\label{lonely-incr}
	After each successful color extension, at most $8$ social edges are turned into non-social ones.
\end{corollary}
\begin{proof}
	Consider any successful color extension, where an uncolored edge $(v, u)$ has obtained a color. By \Cref{base}, a social edge $(v_1, u_0)\in S$ becomes no longer social after a successful iteration if one of the following conditions holds.
	\begin{itemize}[leftmargin=*]
	\item Vertex $v_1$ is incident on the newly colored $(v, u)$ which happens to be assigned the same color $\chi(v, u) = \clr_\chi(v_1, u_0)$, or $v_1$ is the endpoint of the maximal alternating path $P$ or $Q$.
	
	In this case, there are at most $4$ choices for this possible social edge $(v_1, v_0)$.
	
	\item There exists a unique social edge $(v_2, u_0)\neq (v_1, u_0)$ such that $\clr_\chi(v_1, u_0) = \clr_\chi(v_2, u_0)$, and $(v_2, u_0)$ is colored or $v_2$ is the endpoint of the maximal alternating path $P$ or $Q$.
	
	In this case, there are also at most $4$ choices for this possible social edge $(v_1, v_0)$.
	\end{itemize}
	Therefore, overall, at most $8$ social edges would no longer be social after a successful extension.
\end{proof}

At any moment, let $m_\rdy, m_\scl, m_\ind, m_\lon$ be the total number edges in $S$ that are ready, social, independent, and lonely, respectively, according to \Cref{def-social}, and therefore $|S| =  m_\rdy + m_\scl + m_\ind + m_\lon$. Define a potential function $\Phi = 10|S| + m_\lon$. Consider any single iteration of random color extension which picks a random edge $(v, u)\in S$. 

\begin{lemma}\label{analysis:ready}
	If $m_\rdy\geq |S|/4$, then with probability $\geq 1/4$, the edge $(v, u)$ gets colored in Step \ref{ext:direct}, and thus in this case the potential $\Phi$ drops by at least $2$.
\end{lemma}
\begin{proof}
	This is straightforward according to our algorithm description.
\end{proof}

\begin{lemma}\label{analysis:social}
	If $m_\scl\geq |S|/4$, then with probability $\geq \frac{3}{80}$, the edge $(v, u)$ gets colored in Step \ref{ext:social}, and thus in this case the potential $\Phi$ drops by at least $2$.
\end{lemma}
\begin{proof}
	For each social edge $e = (v, u)$ and color $x\in \miss_\chi(u)$, define $P_x(e)$ to be the maximal $\{x, y\}$-alternating path starting at vertex $v$ where $y = \clr_\chi(v, u)$.
 
Note that for any fixed $u\in W$ and $x\in \miss_\chi(u)$, the number of different paths $P_x(v, u)$ 
over social edges $(v,u)$ incident on $u$
is precisely the number of such edges, since any two distinct edges lead to different paths. Fixing an arbitrary missed color 
$x\in \miss_\chi(u)$ for each vertex $u \in W$,
the total number of different paths 
over social edges incident on each vertex of $W$ with respect to its arbitrarily fixed missed color
is at least $\frac{1}{2}\cdot m_\scl$, where the factor of $1/2$ stems from the fact that the same path may be counted twice, as two different social edges $e,'e$ could be the endpoints edges of the same alternating path.
  Since $|\miss_\chi(u)|\geq d/2$ for each $u\in W$ and as every different missing color $x$ at $u$ leads to distinct maximal alternating paths, the collection $\mathcal{P}$ of all different paths $P_x(e)$, over all 
  $m_\scl$ social edges $e = (v,u), u\in W$ and over all $x \in \miss_\chi(u)$, satisfies: \begin{equation} \label{pathslb}
  |\mathcal{P}| ~\ge~ \frac{1}{2}\cdot\frac{d}{2}\cdot m_\scl ~=~ \frac{d m_\scl}{4}.
  \end{equation}
	
	Define $\mathcal{P'}$ to be the sub-collection of $\mathcal{P}$, which includes all paths in $\mathcal{P}$ that do not terminate at vertex $u$. We argue that $|\mathcal{P'}| \ge |\mathcal{P} / 2$, which by \Cref{pathslb} yields
 $|\mathcal{P'}| \geq \frac{d m_\scl}{8}$. In fact, for any $u\in W$, any color $x\in \miss_\chi(u)$ and any other color $y$, let $v_1, v_2, \ldots, v_{k\geq 2}$ be all uncolored neighbors of $u$ such that $(v_i, u)\in S$ and $\clr_\chi(v_i, u) = y$. Then, there exists at most one index $1\leq i\leq k$, such that the maximal $\{x,y\}$-alternating path $P_x(v_i, u)$ starting at vertex $v_i$ terminates at $u$ (otherwise $u$ would be incident on more than one edge colored $y$). Therefore, all other maximal $\{x, y\}$-alternating paths $P_x(v_j, u), j\neq i$ do not end at $u$. 
It follows that $|\mathcal{P'}| \ge |\mathcal{P} / 2$.
	
	\Cref{sum-alt-path} implies that the total length of all maximal alternating paths whose lengths are at least $2$ is at most $3(\Delta+1) m$. Therefore, the average length of paths in $\mathcal{P'}$ is at most
    $$\frac{24(\Delta+1) m}{dm_\scl}\leq \frac{128(\Delta+1) m}{d\lambda} < 0.7L,$$
from which we conclude that the expected length of a path, conditioned on the event that
$(v, u)\in \mathcal{P}'$, is bounded by $0.7 L$. 
 Using Markov's inequality, the edge $(v, u)$ gets colored with probability $\geq 0.3$, conditioned on the event that $(v, u)\in \mathcal{P}'$. Therefore, the overall success probability of Step \ref{ext:social} is at least $0.3\cdot \frac{1}{8} = \frac{3}{80}$.
\end{proof}

\begin{lemma}\label{analysis:independent}
	If $m_\ind\geq |S|/4$, then with probability $\geq 0.15$, the edge $(v, u)$ gets colored in Step \ref{ext:vizing}, and thus in this case the potential $\Phi$ drops by at least $2$.
\end{lemma}
\begin{proof}
	For each independent edge $e = (v, u)$ and color $x\in \miss_\chi(u)$, define $Q_x(e)$ to be the maximal $\{x, z\}$-alternating path starting at vertex $u$ if we try to perform Vizing's procedure to extend $\chi$ to $e$ around the colored neighborhood $N_\chi(u)$, where vertex $v$ uses the missing color $\clr_\chi(v, u)$, vertices $w$ use missing colors $c_\chi(w), \forall w\in N_\chi(u)$, and $u$ uses the missing color $x\in \miss_\chi(u)$.
	
	We argue that for any fixed $u\in W$ and $x\in \miss_\chi(u)$, all the maximal alternating paths $Q_x(\cdot, u)$ that correspond to different independent edges are different. Indeed, for any pair of independent edges $(v_1, u), (v_2, u)$, denoting by $(u, w_1)$
and $(u, w_2)$
 the first edges of maximal alternating paths $Q_x(v_1, u)$ and $Q_x(v_2, u)$, respectively, then by definition of independent edges \cref{def-social}, $w_1, w_2$ should belong to different weakly connected components of the digraph $F_\chi(u)$, and therefore $w_1\neq w_2$.
	
	Since $|\miss_\chi(u)|\geq d/2$ for each $u\in W$, the
 collection $\mathcal{Q}$ of all different paths $Q_x(e)$, over all 
  $m_\ind$ independent edges $e = (v,u), u\in W$ and over all $x \in \miss_\chi(u)$,
 satisfies:
	$$|\mathcal{Q}| ~\ge~ \frac{1}{2}\cdot\frac{d}{2}\cdot m_\ind ~=~ \frac{d m_\ind}{4},$$
where the extra factor of $1/2$ stems from the fact that two different independent edges $e, e'$ could be the endpoints edges of the same alternating path.

\Cref{sum-alt-path} implies that the total length of paths in $\mathcal{Q}$ whose lengths are at least $2$ is at most $3(\Delta+1) m$. Therefore, the average length of paths in $\mathcal{Q}$ is at most
	$$\frac{12(\Delta+1) m}{dm_\ind} ~\leq~ \frac{64(\Delta+1) m}{d\lambda} ~<~ 0.4L,$$
from which we conclude that the expected length of a path, conditioned on event that edge $(v,u)$ is independent, is bounded by $0.4 L$.
	Using Markov's inequality, the edge $(v, u)$ gets colored with probability $\geq 0.6$, conditioned on the event that edge $(v, u)$ is independent. Therefore, the overall success probability of Step \ref{ext:social} is at least $0.15$.
\end{proof}

\begin{lemma}\label{analysis:lonely}
	If $m_\lon \geq |S|/4$, then with probability $\geq 0.25$, the potential function $\Phi$ drops by $2$ after Step \ref{ext:pairing}.
\end{lemma}
\begin{proof}
	Since $m_\lon \geq |S|/4$, with probability at least $0.25$, the random edge $(v, u)$ would be lonely. In this case, we would execute Step \ref{ext:pairing} which turns two lonely edges $(v, u), (v', u)$ into social ones. Thus, $\Phi$ decreases by $2$.
\end{proof}

\subsection{Proof of \Cref{lem:key 1}} \label{wrapup} 

By \Cref{analysis:ready},  \Cref{analysis:social}, \Cref{analysis:independent}, and \Cref{analysis:lonely}, each iteration of the random color extension procedure decreases $\Phi$ by at least $2$ with constant probability. Since the initial potential value of $\Phi$ is set as $10|S| + m_\lon \le 11\lambda$, it follows that the algorithm terminates after $O(\lambda\log n)$ iterations
with high probability. By \Cref{runtime}, each iteration takes time  $\tilde{O}(\Delta + L)$, hence the total running time is bounded by $
\tilde{O}(\lambda \log n(\Delta + L)) = 
\tilde{O}\brac{\Delta \lambda + \frac{\Delta m}{d}}$
with high probability. 

As for the total number of newly colored edges, $|S|$ decreases either when a new edge in $S$ gets colored or when a vertex $u\in U^\star$ is removed from $W$. In the latter case, when a vertex $u\in U^\star$ is removed from $W$, since $|\miss_\chi(u)|$ was at least $d$ at the beginning, there must have been at least $d/2$ newly colored edges incident on $u$ when $u$ has left $W$. Consequently, we can charge the loss of at most $d/2$ units to the size of $|S|$ due to the removal of $u$ from $S$ to the at least $d/2$ newly colored edges incident on $u$; however, a newly colored edge may be charged by its two endpoints. In other words, if the size of $S$ has reduced by $\delta$ units, at least $\delta/3$ edges of $S$ have been colored. Finally, noting that our algorithm terminates whenever the size of $S$ drops below $\frac{3\lambda}{4}$, it follows that at this stage we
must have colored at least $\frac{\lambda}{12}$ new edges, which concludes the proof of \Cref{lem:key 1}.

%% file: ack.tex
\section*{Acknowledgements}

Shay Solomon is funded by the European Union (ERC, DynOpt, 101043159). Views and opinions expressed are however those of the author(s) only and do not necessarily reflect those of the European Union or the European Research Council. Neither the European Union nor the granting authority can be held responsible for them. Shay Solomon is also supported by the Israel Science Foundation (ISF) grant No.1991/1, and by a grant from the United States-Israel Binational Science Foundation (BSF), Jerusalem, Israel, and the United States National Science Foundation (NSF).